\documentclass[
aps,
prx,
twocolumn,
superscriptaddress,
showpacs,
final,
floatfix,
longbibliography,
]{revtex4-2}
\usepackage[utf8]{inputenc}
\usepackage[toc,page]{appendix}
\usepackage{amsfonts}
\usepackage[dvips]{graphicx}
\usepackage{amsmath}
\usepackage{amssymb}
\usepackage{hyperref}
\usepackage{color}
\usepackage{mathrsfs}
\usepackage{isomath}
\usepackage{amsthm}
\usepackage{epstopdf}
\usepackage{txfonts}
\usepackage{dsfont}
\usepackage{ulem}
\usepackage{physics}
\allowdisplaybreaks[4]
\usepackage{bbold}
\usepackage{mathtools}
\usepackage{xcolor}
\usepackage{bm}
\usepackage{orcidlink}
\usepackage{algorithm}
\usepackage{algpseudocode}
\usepackage{listings}
\usepackage{lipsum}

\renewcommand{\tr}{\mathrm{tr}}

\renewcommand{\emph}[1]{{\it #1}}

\newcommand{\be}{\begin{equation}}
\newcommand{\ee}{\end{equation}}

\def\bea#1\eea{%
  \begin{align}%
  #1%
  \end{align}%
}

\newcommand{\id}{\mathbb{1}}
\newcommand{\idmap}{{\rm id}}          

\newcommand{\muct}{\mu_{\rm ct}}       
\newcommand{\mudt}{\mu_{\rm dt}}       
\newcommand{\sigmact}{\sigma_{\rm ct}} 
\newcommand{\sigmadt}{\sigma_{\rm dt}} 

\newcommand{\ZZ}{\mathcal Z}

\providecommand{\varrho}{{\mathrm P}}

\newtheorem{theorem}{Theorem}
\newtheorem*{theorem*}{Theorem}

\usepackage{cleveref}
\crefname{equation}{Eq.}{Eqs.}
\crefname{figure}{Fig.}{Figs.}
\crefname{observation}{Obs.}{Obs.}
\crefname{corollary}{Corollary}{Corollaries}
\crefname{lemma}{Lemma}{Lemmata}
\crefname{proof}{Proof}{Proofs}
\creflabelformat{proof}{#2proof#3}
\crefname{remark}{Remark}{Remarks}
\crefname{prop}{Proposition}{Propositions}
\crefname{appendix}{Appendix}{Appendix}

\AtBeginEnvironment{appendices}{\crefalias{section}{appendix}}

\usepackage{hyperref}
\hypersetup{citecolor=blue}
\hypersetup{colorlinks=true}
\hypersetup{linkcolor=blue}
\hypersetup{urlcolor=blue}

\definecolor{BO}{RGB}{105, 0, 255}

\begin{document}
\title{Least Variable Quantum Counting Processes}

\author{Bita Olamaei\,\orcidlink{0009-0008-4275-5117}}
\thanks{Equal contributions}
\email{bita.olamaei@tuwien.ac.at}
\affiliation{Atominstitut, Technische Universit\"{a}t Wien, Stadionallee 2, 1020 Vienna, Austria}
\affiliation{Vienna Center for Quantum Science and Technology, Technische Universit\"{a}t Wien, 1020 Vienna, Austria}

\author{Florian Meier\,\orcidlink{0000-0003-4337-3846}}
\thanks{Equal contributions}
\email{florianmeier256@gmail.com}
\affiliation{Atominstitut, Technische Universit\"{a}t Wien, Stadionallee 2, 1020 Vienna, Austria}
\affiliation{Vienna Center for Quantum Science and Technology, Technische Universit\"{a}t Wien, 1020 Vienna, Austria}
\affiliation{Institut für theoretische Physik, Technische Universit{\"a}t Wien, 1040 Vienna, Austria}

\author{Costantino Budroni\,\orcidlink{0000-0002-6562-7862}}
\email{costantino.budroni@unipi.it}
\affiliation{Department of Physics “E. Fermi”, University of Pisa, Largo B. Pontecorvo 3, 56127 Pisa, Italy}

\author{Pharnam Bakhshinezhad\,\orcidlink{0000-0002-0088-0672}}
\email{pharnam.bakhshinezhad@tuwien.ac.at}
\affiliation{Atominstitut, Technische Universit\"{a}t Wien, Stadionallee 2, 1020 Vienna, Austria}
\affiliation{Vienna Center for Quantum Science and Technology, Technische Universit\"{a}t Wien, 1020 Vienna, Austria}

\author{Giuseppe Vitagliano\,\orcidlink{0000-0002-5563-3222}}
\email{giuseppe.vitagliano@tuwien.ac.at}
\affiliation{Atominstitut, Technische Universit\"{a}t Wien, Stadionallee 2, 1020 Vienna, Austria}
\affiliation{Vienna Center for Quantum Science and Technology, Technische Universit\"{a}t Wien, 1020 Vienna, Austria}

\begin{abstract}
    Counting processes provide a fundamental description of stochastic events ranging from photon detection to clock ticks. 
    A central question is how accurately such events can be timed when only finite memory resources are available. 
    Here, we investigate this problem within a general framework of finite-dimensional classical and quantum counting processes. 
    We derive a rigorous finite-memory variance bound obeyed by every classical $d$-state counting process, which is tight and 
    saturated by a discrete Erlang-type ladder process. 
    Through numerical optimization, we identify quantum counting processes that violate this classical bound, achieving smaller 
    first-tick fluctuations than any classical process with the same memory size and mean tick time. 
    For the qubit case, we further derive an analytical large-mean bound within a single-Kraus no-tick family, showing that the 
    quantum advantage persists asymptotically within this class. 
    The optimized quantum processes exhibit coherent conditioned dynamics and approach a continuous-time quantum-jump description 
    as the mean increases. 
    Our results establish a finite-memory quantum advantage in temporal precision and connect discrete-time counting processes 
    with continuous-time quantum timekeeping.
\end{abstract}

\maketitle 

\section{Introduction}
\label{sec:intro}

\begin{figure*}
    \centering
    \includegraphics[width=\linewidth]{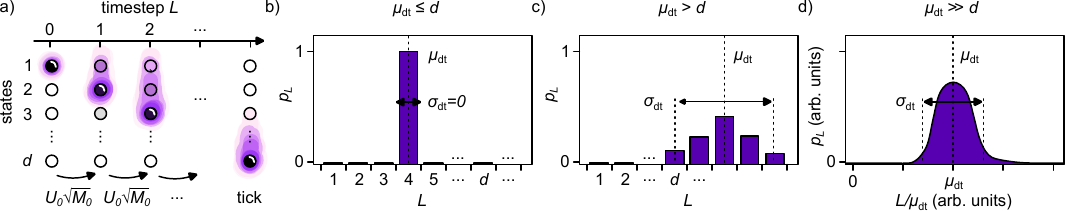}
    \caption{Conceptual overview of the setup.
    (a) Discrete-time evolution of a finite-dimensional quantum system is considered.
    The waiting time probability distribution is shown as a function of the number of timesteps $L$ in the panels b-d.
    (b) For memory size $d$, any number of counts up and including $d$ can be deterministically counted.
    (c) For a mean tick time $\mudt$ strictly greater than $d$, the tick distribution necessarily becomes probabilistic.
    (d) In the limit of $\mudt\rightarrow\infty$, the discrete-time process converges to a continuous-time one with continuous waiting time distribution.}
    \label{fig:ConceptualPanel}
\end{figure*}

The simplest counting machine has access to only a single bit of memory. 
Such a machine can {\it perfectly count} to at most two, i.e., perfectly distinguish two successive stages. 
A first step flips the bit from $0$ to $1$, the second step returns the bit back to $0$, making the machine output a \textit{tick} signaling the completion of a two-step cycle. 
Deterministically counting three or more steps is, however, impossible with a single bit. 
Nevertheless, the same machine can produce a tick after three steps on average by advancing from one stage to the next only probabilistically, for example, with probability $2/3$ at each step.
Completing each stage, then, requires on average $3/2$ steps, thereby, producing a tick, on average, after three steps. 
The price for producing longer waiting times with the same finite memory is that the tick time becomes stochastic: as the number of steps to be counted increases, its probability distribution necessarily broadens. 
Achieving a sharper tick distribution, and hence higher precision, therefore requires additional memory. 
To meaningfully quantify this internal memory, we consider {\it autonomous machines}, whose dynamics do not depend on an external input that could itself encode information about the elapsed time \cite{Erker2017,Malabarba2015,Woods2021,Woods2022,Silva2023,Campbell2026}.

Counting processes arise in a broad range of applications, including molecular mechanics~\cite{Barato2017,Duso2020}, financial stochastic processes~\cite{Cohen2015}, and queuing theory~\cite{Erlang1917,cox1962renewal,Whitt1981}. 
From a more fundamental perspective, counting processes also provide a framework for addressing the question of distinguishing classical from quantum temporal correlations, 
for example through extensions of the notion of {\it macrorealism}~\cite{leggett_quantum_1985,leggett_realism_2008,emary_leggettgarg_2013}
to hidden variable models with bounded disturbance~\cite{budroni2019_memorycost,Budroni2021Ticking-clock,Vitagliano2023}.

Here, we use internal memory as the central resource constraining the precision of the counting process, which we quantify using the variance $\sigmadt^2$ and mean $\mudt$ of its tick distribution (\cref{fig:ConceptualPanel}). 
For classical counters, it was conjectured that the variance is tightly lower-bounded by memory size and mean~\cite{Budroni2021Ticking-clock}. 
Here, we resolve this conjecture and rigorously prove that for an arbitrary memory dimension $d$, the variance obeys
\begin{equation}
    d \sigmadt^2 \geq \mudt (\mudt - d).
    \label{eq:main_result_intro}
\end{equation}
The bound is tight, and for $\mudt \geq d$, is saturated by a $d$-stage discrete Erlang ladder~\cite{Erlang1917,Cox1970,Getz2018}.

Quantum counters, in contrast, generate a temporal sequence of events by using quantum memory. 
Such systems have been investigated particularly in connection with timekeeping~\cite{Renner2017,Erker2017,Woods2019, Woods2022} and computation~\cite{Malabarba2015,Xuereb2023,MarinGuzman2024,Woods2024,Guzman2025} as well as with more general output sequences~\cite{budroni2019_memorycost,spee2020_simulating,Budroni2021Ticking-clock,vieira2022_temporal,Vitagliano2023,Weilenmann2024}. 
Given the same memory size, a quantum system can, in principle, attain a lower variance in the distribution of its ticks than a classical one~\cite{Woods2022,Budroni2021Ticking-clock}.
These results motivate the question of how the precision of counting processes in quantum systems is fundamentally constrained by the available memory, and how these constraints differ from the classical case.

To address this question, we develop a systematic optimization method for finite-dimensional quantum counting processes. 
Using this method, we numerically identify explicit quantum models that violate the classical variance bound, demonstrating a quantum advantage at a fixed memory dimension. 
We further analyze the structure of these optimized processes, comparing their dynamics with those of the classical Erlang process, which saturates the classical bound. 
Finally, for the two-dimensional case, we show that the optimization problem can be expressed in closed form.

In the large-mean regime, these optimized discrete-time quantum counting processes approach a continuous-time quantum master equation description~\cite{Budroni2021Ticking-clock,Woods2021}, which establishes the connection to autonomous quantum clocks~\cite{Erker2017,Woods2019,Woods2021,Woods2022,Silva2023}. 
In this limit, our classical result~\eqref{eq:main_result_intro} recovers the established linear bound on clock precision in terms of the memory dimension~\cite{AldousShepp1987,Woods2022}. 
For quantum counting processes, our numerical optimization reveals a superquadratic scaling of precision across the range of dimensions investigated. 
This behavior suggests a transient, finite-dimensional correction to the quadratic quantum bound derived in Ref.~\cite{Yang2020}.

Together, these results establish finite memory as a fundamental constraint on temporal precision and demonstrate how quantum dynamics can surpass classical limitations. 
Our findings thereby extend previously established bounds on precision in quantum clocks~\cite{Woods2022}. 
Memory thus complements constraints involving other resources---such as temporal resolution~\cite{Schwarzhans2021,Meier2023,Prech2025} and entropy production~\cite{Dost2023,Meier2025a}---that have been established in different settings. 

\section{Quantum counting processes: Framework}
\label{sec:count-process}
To compare classical counting processes with their quantum counterparts, we work within the framework of \textit{finite-state automata}~\cite{Kondacs1997,Say2014}.
Finite-state automata are time-homogeneous sequential devices that take an input from a set $\mathcal{X}$, update their internal state, and produce an output in a set $\mathcal{A}$.

Specifically, a \textit{classical automaton} of dimension $d$ is specified by an initial probability vector $\pi \in \mathbb{R}^d$ and a family of substochastic transition matrices $\{ T(a|x) \}_{a,x}$. Their entries satisfy
\begin{equation} 
    [T(a|x)]_{jk}\ge 0, 
    \qquad \sum_{a,k}[T(a|x)]_{jk} = 1 
    \quad \forall j,x,
\end{equation}
so that $\sum_a T(a|x)$ is a stochastic matrix (using right-stochastic matrices by convention).
Writing $\eta = (1, \dots, 1)^{\mathsf T}$, the probability of observing a sequence $a_1, \dots, a_n$ under inputs $x_1, \dots, x_n$ is
\begin{equation}
    p(a_1 \cdots a_n|x_1 \cdots x_n) = \pi^{\mathsf T} \, T(a_1|x_1) \cdots T(a_n|x_n) \,\eta.
\end{equation}

The analogous quantum model~\cite{Say2014,hoffmann2018,Arrighi2019,budroni2019_memorycost, vieira2022_temporal}, is specified by an initial density operator $\varrho_{\rm in}$ on a $d$-dimensional Hilbert space and, for each input $x$, a quantum instrument $\{\mathcal{I}_{a|x}\}_a$, i.e., a collection of completely positive maps such that $\sum_a \mathcal{I}_{a|x}$ is completely positive and trace-preserving (CPTP). 
The probability of a sequence of outputs is given by the Born rule:
\begin{equation}
    p(a_1 \cdots a_n|x_1 \cdots x_n) = \mathrm{tr} \left[ \mathcal{I}_{a_n|x_n} \circ \cdots \circ \mathcal{I}_{a_1|x_1} (\varrho_{\rm in}) \right].
\end{equation}

\subsection{Discrete-Time Counting Processes}
\label{sbsec:dt-count-process}
To model an autonomous discrete-time counting process, we consider the simplest counting scenario. 
To this end, we assume $\mathcal{X}$ is trivial (there is no external input) and limit the outputs to $\mathcal{A} = \{0, 1\}$, corresponding to the absence or presence of a tick, respectively. 
Moreover, the same two-outcome operation is applied at each time step. 
We denote the corresponding classical transition matrices by $\{T_0, T_1 \}$ and the quantum instrument by $\{\mathcal{I}_0, \mathcal{I}_1\}$. 

A complete clock model generally requires specifying both the no-tick ($0$) and tick ($1$) branches of the instrument, as well as the state of the system immediately after a tick event.
Throughout this work, we consider clocks that operate cyclically, meaning that each tick reinitializes the system to the same reference initial state $\varrho_\mathrm{in}$. 
With this reset mechanism, successive inter-tick intervals are governed by the same underlying dynamics, and the resulting counting process belongs to the class of renewal processes~\cite{cox1962renewal}. 
In this case, the long-time counting statistics are completely determined by the distribution of a single inter-tick interval.
For this reason, it is sufficient to focus on the first-tick statistics generated by the no-tick evolution $\mathcal{I}_0$, which will be the main object of study in what follows.
See \cref{fig:ConceptualPanel} for a conceptual overview.

In a discrete-time scenario, starting from an initial state $\varrho_\mathrm{in}$, the probability that the first tick occurs at time step $L \in \mathbb N$ is
\begin{equation}
    p_L := \tr \left[ \mathcal{I}_1^{ } \circ \mathcal{I}_0^{L-1}(\varrho_\mathrm{in}) \right] 
         = \tr \left[ (\id - M_0) \mathcal{I}_0^{L-1}(\varrho_\mathrm{in}) \right].
    \label{eq:p(L)_def}
\end{equation}
The sequence of such probabilities forms the \textit{first-tick-time} distribution. 
Here, $\id - M_0 = M_1 = \mathcal{I}_1^*(\id)$ is the positive operator-valued measure (POVM) element associated with the tick event, and $M_0$ is the one associated with no-tick.
Mean and variance of the tick time can be obtained as
\begin{align}
    \mudt &= \sum_{L\geq 1} L p_L, \\ 
    \sigmadt^2 &= \sum_{L\geq 1} (L - \mudt)^2 p_L.
\end{align}

A canonical realization of the reset mechanism discussed above is given by the measure-and-prepare map
\begin{equation}
    \mathcal{I}_1(\varrho) = \tr \left[ (\id - M_0) \varrho \right] \, \varrho_\mathrm{in},
    \label{eq:reset_map}
\end{equation}
which prepares the reference state $\varrho_\mathrm{in}$ after each tick event. 
The corresponding unconditional evolution reads
\begin{equation}
    \Phi(\varrho) = \mathcal{I}_0(\varrho) + \tr \left[ (\id - M_0) \varrho \right] \, \varrho_\mathrm{in}.
\end{equation}
For the reset map \cref{eq:reset_map}, the waiting times $\tau_1, \tau_2, \ldots$ between consecutive ticks are independent and identically distributed random variables with a common distribution
\begin{equation}
    \mathrm{Prob}(\tau_k = L) = p_L, 
    \qquad \forall\, k \geq 1.
\end{equation}
The time of the $n$th tick is thus given by the partial sum $T_n = \tau_1 + \cdots + \tau_n$, formally establishing the resulting sequence of events as a renewal process~\cite{cox1962renewal}.

Using $\Phi = \mathcal{I}_0 + \mathcal{I}_1$ and the trace-preserving property of $\Phi$, we can rewrite \cref{eq:p(L)_def} in the convenient form
\begin{equation}
    p_L = f_{L-1} - f_L,
    \label{eq:pL_fL}
\end{equation}
where 
\begin{equation}
    f_L := \tr \left[ \mathcal{I}_0^{L} (\varrho_\mathrm{in}) \right]
\end{equation}
is the probability that no tick occurs in the first $L$ steps. Note that it satisfies $f_0 = \tr(\varrho_\mathrm{in}) = 1$. 

The conditional probability that a tick occurs at the next step, given survival up to step $L-1$, is
\begin{equation}
    h_L := \frac{p_L}{f_{L-1}} = 1 - \frac{f_{L}}{f_{L-1}} = \tr \left[ (\id - M_0)\, \varrho_{L-1} \right],
    \label{eq:dt_hazard_def}
\end{equation}
where $\varrho_{L-1}:=\mathcal I_0^{L-1}(\varrho_\mathrm{in})/{f_{L-1}}$ is the normalized conditional state after $L-1$ no-tick time steps (defined whenever $f_{L-1} > 0$).
We refer to $h_L$ as the \textit{discrete-time hazard} sequence~\cite{Muthen2005}.

\subsection{Continuous-Time Limit} 
\label{subsec:ct-model} 
Physically, the elementary no-tick step $\mathcal{I}_0$ may arise as the evolution of a quantum system during a time interval of duration $\delta > 0$. We are particularly interested in the continuous monitoring limit, $\delta \rightarrow 0$. 
In this regime, we assume that the no-tick step approaches the identity, $\mathcal{I}_0^{(\delta)} \rightarrow \text{id}$. 
We will explore this case in detail later.

Under the following standard assumptions on the no-tick family $\{\widetilde{\mathcal{I}}_0^{(t)}\}_{t \ge 0}$~\cite{Woods2021}:
\begin{equation}
    \widetilde{\mathcal{I}}_0^{(0)} = \idmap, 
    \qquad \widetilde{\mathcal{I}}_0^{(t+s)} = \widetilde{\mathcal{I}}_0^{(t)} \circ \widetilde{\mathcal{I}}_0^{(s)}, 
    \qquad \lim_{t\to0^+} \bigl\| \widetilde{\mathcal{I}}_0^{(t)} - \idmap \bigr\| = 0,
    \label{eq:ct_semigroup_properties_rewritten}
\end{equation}
the dynamics forms a strongly continuous semigroup and hence admits a generator $\mathcal{L}_0$ such that~\cite{Pazy1983}
\begin{equation}
    \widetilde{\mathcal{I}}_0^{(t)}=e^{t\mathcal{L}_0}.
    \label{eq:ct_notick_semigroup_rewritten}
\end{equation}
At the infinitesimal level, the generator is obtained from the discrete no-tick map as
\begin{equation}
    \mathcal{L}_0 = \lim_{\delta\to0^+} \frac{1}{\delta}(\mathcal{I}_0^{(\delta)} - \idmap).
    \label{eq:ct_generator_limit_rewritten}
\end{equation}

Therefore, in continuous time, the quantum counting process may be directly modeled using the master equation in Lindblad form by defining the \textit{ticks} of the clock as the \textit{quantum jumps} of the master equation. 
At the differential level, the time evolution of a quantum state $\varrho$ under such a Lindblad equation is given by
\begin{equation}
    \dot\varrho = -i[H,\varrho] + \sum_\alpha \left( L_\alpha^{ } \varrho L_\alpha^\dagger - \frac{1}{2}\{L_\alpha^\dagger L_\alpha^{ } ,\varrho\}\right),
    \label{eq:dot_rho_lindblad}
\end{equation}
which is a typical formalism to describe the effective evolution of a quantum system weakly coupled to an environment~\cite{Breuer2007,Hofer2017}.
Under a stochastic unraveling of this master equation~\cite{Wiseman2009,Landi2024,Bettmann2024,Meier2026}, just as in the discrete-time scenario, we consider the time it takes for the first jump associated with some operator $L_\alpha$ to occur, which is also known as the \textit{waiting time} or the \textit{first passage time} in stochastic systems~\cite{Prech2025,Kewming2024}.

In the most general setting, let us consider a subset $\aleph = \{\alpha : \text{tick} \}$ of jumps that we count as ticks~\cite{Landi2024,Prech2025}.
To determine the time of the first jump associated with an operator $L_\alpha^{ }$ where $\alpha \in \aleph$, we can consider the conditional non-normalized state $\varrho^0(t)$ that follows the evolution equation $\dot \varrho^0 = \mathcal{L}_0 \, \varrho^0$, with~\cite{Dost2023,Meier2025a}
\begin{equation}
    \mathcal{L}_0(\,\cdot\,) = -i[H, \, \cdot \,] - \frac{1}{2} \sum_{\alpha} \left \{ L_\alpha^\dagger L_\alpha^{ }, \, \cdot \, \right \} + \sum_{\beta \notin \aleph} L_\beta^{ } \, \cdot\, L_\beta^\dagger.
    \label{eq:L_0_general}
\end{equation}
The non-trace-preserving operator $\mathcal{L}_0$ is the general form of the conditional no-tick evolution generator in continuous time, as obtained in the limit of \cref{eq:ct_generator_limit_rewritten}.

Taking the trace of $\varrho^0(t)$ results in the survival probability
\begin{equation}
    f(t) := \tr[\varrho^0(t)].
    \label{eq:f(t)_CT_def}
\end{equation}
This represents the probability that, at time $t$, a jump has not yet occurred. Since the operators $L_\alpha^\dagger L_\alpha$ are positive semi-definite matrices ($L_\alpha^\dagger L_\alpha\geq 0$), this probability is monotonically decreasing.
Given the initial condition $\tr[\varrho^0(0)]=1$, $f(t)$ describes a well-defined complementary cumulative probability. 
The time derivative then gives the corresponding probability density $p(t)=-\partial_t f(t)$, which characterizes the waiting-time distribution. 
Using $p(t)$, the mean and variance of the waiting time can be determined as:
\begin{align}
    \muct &= \int_0^\infty \dd t\, t \, p(t),\\ 
    \sigmact^2 &= \int_0^\infty \dd t\, (t-\muct)^2 p(t).
\end{align}

The continuous-time model can be related back to the discrete-time model on the level of the tick probability distribution.
Using the identity in \cref{eq:pL_fL} in the limit as $\delta \rightarrow 0^+$, the following holds up to corrections of order $\delta^2$:
\begin{equation}
    p_L = \delta\, p(L\delta), 
    \qquad (\text{as }\delta\rightarrow 0^+).
    \label{eq:disc_to_cont_pL}
\end{equation}
Moreover, in the limit as $\delta\rightarrow 0^+$, sums over the discrete-time steps converge to Riemann integrals, as $\delta \sum_{L\geq 1} \rightarrow \int_0^\infty \dd t$, yielding:
\begin{align}
    \delta \, \mudt &\xrightarrow{\delta \rightarrow 0^+}  \muct, \text{ and} 
    \label{eq:mudt_to_muct} \\
    \delta^2 \,\sigmadt^2  &\xrightarrow{\delta \rightarrow 0^+} \sigmact^2.
    \label{eq:sigmadt_to_sigmact}
\end{align}

As in the discrete-time case, a hazard rate can be defined to quantify the instantaneous tick probability rate~\cite{Allison1982,Kleinbaum2012}. 
Given $\varrho^\text{no-tick}(t) = \varrho^0(t)/f(t)$, the normalized conditional no-tick state, the hazard rate is defined analogously to \cref{eq:dt_hazard_def} as:
\begin{equation}
    h(t) := -\tr \left[ \mathcal{L}_0 \, \varrho^\text{no-tick}(t) \right] = \sum_{\alpha \in \aleph} \tr \left[ L_\alpha^\dagger L_\alpha^{ } \, \varrho^\text{no-tick}(t) \right].
    \label{eq:h(t)_def}
\end{equation}

\section{Figures of merit and classical finite-memory bound} 
\label{sec:figures_of_merit_classical_bound}
An accurate counting process should produce its first tick within a narrow temporal window around a fixed mean time. 
The corresponding optimization problem is therefore to minimize the variance for fixed mean, subject to a finite internal memory of dimension $d$. 
The constraint on the system's internal memory is essential for distinguishing classical from quantum counting processes \cite{hoffmann2018,budroni2019_memorycost,spee2020_simulating}. 
Without it, an arbitrarily large classical memory can encode the elapsed time and generate an arbitrarily sharp first-tick distribution. 
Thus, we compare quantum processes acting on a $d$-dimensional Hilbert space to classical automata with $d$ internal states.

\subsection{Classical Finite-Memory Variance Bound}
\label{subsec:classical_variance_benchmark}
In discrete time, for $\mudt \leq d$, a classical counter can produce a deterministic first tick and hence achieve $\sigmadt^2 = 0$.
Contrarily, for $\mudt > d$, a perfect deterministic ticking mechanism is not possible, and the tick variance has a nontrivial lower bound.
A lower bound on $\sigmadt^2$ was only proven for dimension $d=2$~~\cite{Budroni2021Ticking-clock}, but for higher $d$, the lower bound remained a conjecture.
Here, we rigorously prove that for any classical $d$-state discrete-time automaton, the first-tick mean $\mudt$ and variance $\sigmadt^2$ obey the following theorem:
\begin{theorem}[Classical finite-memory variance bound]
\label{thm:classical_variance_bound}
    The variance of any $d$-dimensional classical discrete-time counting process, with arbitrary initial distribution and tick transition that eventually ticks, is bounded by,
    \begin{equation}
        d\,\sigmadt^2 \geq \mudt(\mudt-d).
        \label{eq:classicalbench}
    \end{equation}
\end{theorem}

The lower bound was known to be saturated by the Erlang-type ladder process described in \cref{subsec:erlang_ladder_clock}, but the optimality of this process in discrete time remained unknown.
This work provides the first rigorous proof of the validity of \cref{eq:classicalbench} for arbitrary dimension $d$.
\begin{proof}[Proof idea]
    The central idea of the proof is to map the discrete-time process to a continuous-time one via Poissonization~\cite{Serfozo2009}.
    In the continuous-time representation, a well-known results by Aldous and Shepp proves optimality of the (continuous-time) Erlang process~\cite{AldousShepp1987}.
    We map this optimality to our setup and thereby establish \cref{eq:classicalbench}, our first main result.
    The full statement of the theorem, including detailed proof, is presented in \cref{app:classical_bound}.
\end{proof}

Motivated by \cref{eq:classicalbench}, we define the discrete-time cost function as follows:
\begin{equation}
    \mathrm{CF}_d := d \, \sigmadt^2 - \mudt(\mudt-d).
    \label{eq:cost_function}
\end{equation}
Since all classical $d$-state processes are bound by \cref{eq:classicalbench}, they inherently satisfy $\mathrm{CF}_d \geq 0$. 
Consequently, for a fixed dimension and mean waiting time
\begin{equation}
    \mathrm{CF}_d < 0
\end{equation}
rules out every classical $d$-state realization of the same first-tick mean and variance.

\subsection{Continuous-Time Limit and Relative Precision}
\label{subsec:discrete_FOM_hazard_bounds}
Let \(\delta\) denote the duration of one discrete time step. 
In the continuous-time limit, multiplying the discrete-time cost function by \(\delta^2\) gives [cf.~Eqs.~\eqref{eq:mudt_to_muct} and~\eqref{eq:sigmadt_to_sigmact}],
\begin{align}
    \delta^2\mathrm{CF}_d \xrightarrow{\delta\rightarrow0^+} d\,\sigmact^2-\muct^2.
    \label{eq:CF_ct_limit}
\end{align}
Thus, the continuous-time counterpart of \(\mathrm{CF}_d<0\) is the condition
\begin{equation}
    \muct^2>d\,\sigmact^2.
\end{equation}

The overall time scale of a continuous-time process is arbitrary. 
A uniform rescaling of the no-tick generator,
\begin{equation}
    \mathcal L_0 \longmapsto \alpha\mathcal L_0, 
    \qquad \alpha > 0,
\end{equation}
induces $\muct\to{\muct}/{\alpha}$, and $\sigmact^2\to{\sigmact^2}/{\alpha^2}$.
The limiting expression in \cref{eq:CF_ct_limit} therefore depends on the arbitrary choice of time unit. 
The natural scale-independent figure of merit is instead the relative precision~\cite{Stupar2018,Erker2017,Woods2021,Wadhia2025}
\begin{equation}
    \mathcal N := \frac{\muct^2}{\sigmact^2}.
    \label{eq:N_precision_Def}
\end{equation}
Maximizing \(\mathcal N\) is equivalent to minimizing \(\sigmact^2\) at any fixed value of \(\muct\)~\cite{Meier2023}.

\begin{figure*}[ht]
    \centering
    \includegraphics[width=\linewidth]{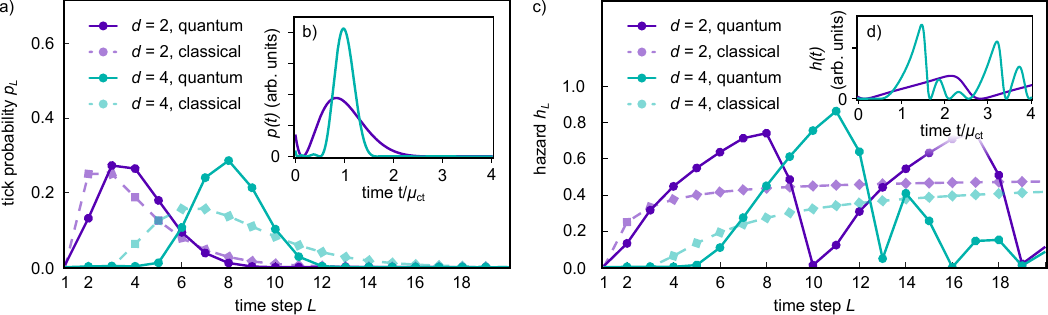}
    \caption{First-tick statistics of an optimized discrete-time quantum counting process. 
    (a) First-tick distribution \(p_L=f_{L-1}-f_L\), shown as a function of the discrete waiting time \(L\). 
    The coherent dynamics concentrates most of the probability in a narrow interval around the prescribed mean \(\mudt\), here chosen to be $\mudt=2d$. 
    (b) The inset shows the continuous-time waiting time distribution $p(t)$. 
    (c) Corresponding hazard \(h_L=1-f_L/f_{L-1}\). 
    The hazard remains suppressed during the initial no-tick evolution and rises sharply when the conditioned state reaches the detector direction. 
    (d) The inset shows the hazard $h(t)$ for continuous-time.}
    \label{fig:optimized_distribution_hazard}
\end{figure*}

At fixed physical mean \(\delta\mudt=\mu_*\), the discrete cost tends to
\begin{equation}
    \delta^2\mathrm{CF}_d \longrightarrow \mu_*^2
    \left(
        \frac{d}{\mathcal N}-1
    \right).
    \label{eq:min_CF_ct}
\end{equation}
The classical continuous-time bound
\begin{equation}
    \mathcal N \leq d
    \label{eq:classical_precision_bound}
\end{equation}
is the continuous-time counterpart of the bound in Eq.~\eqref{eq:classicalbench}.
In continuous time, this inequality is a direct consequence of the Aldous--Shepp bound and the optimality of the Erlang process~\cite{Erlang1917,AldousShepp1987}.
Hence, a process with $\mathcal N>d$ constitutes a violation of the classical continuous-time finite-memory bound. 
Quantum counting processes generically displaying such an advantage have been constructed in Refs.~\cite{Woods2022,Dost2023,Meier2025a}.
However the quantum process that gives the maximal possible value of $\mathcal{N}$ is not known.

\subsection{Hazard Representation of the Cost}
\label{subsec:hazard_bounds}
The cost function also admits a useful representation in terms of the discrete hazard.
We can define the largest hazard attained along the conditioned trajectory as
\begin{equation}
    h_{\max} := \sup_{L\geq1}h_L,
    \label{eq:hmax_def_rewritten}
\end{equation}
which is the largest instantaneous tick probability.
It allows to formulate a general lower bound on the discrete-time cost function as (derivation in \cref{app:derivation_hazard_1})
\begin{equation}
    \mathrm{CF}_d \geq \frac{2d\,\mudt}{h_{\max}} - (d+1)\mudt^2.
    \label{eq:CF_hazard_bound_discrete_rewritten}
\end{equation}

For a quantum counting process with no-tick effect \(M_0\), one also has
\begin{equation}
    h_L = \tr[ (\id-M_0)\varrho_{L-1}] \leq \|\id-M_0\|_\infty.
\end{equation}
The trajectory-dependent quantity \(h_{\max}\) can nevertheless be strictly smaller than the operator-norm bound, since the conditioned orbit need not visit the most detectable state.

Equation~\eqref{eq:CF_hazard_bound_discrete_rewritten} captures a simple tradeoff. 
A process whose hazard remains small throughout the conditioned evolution can achieve a large prescribed mean only by developing a broad survival tail. 
Conversely, a sharply concentrated first-tick distribution requires the hazard to remain suppressed initially and then increase rapidly near the typical ticking time. 
The temporal structure of this rise will play a central role in the optimized quantum processes studied below (see Fig~\ref{fig:optimized_distribution_hazard}). 
It also extends the continuous-time precision-resolution bound from Ref.~\cite{Meier2023} to discrete time.

Consequently, a process whose conditioned trajectory remains only weakly detectable cannot simultaneously produce a sharply concentrated first-tick distribution. 
The bound is generally not expected to be tight, since it retains only the largest value of the hazard and discards its detailed temporal profile, but it provides a simple state-dependent restriction on the achievable cost.

\section{Optimized counting processes}
\label{sec:optimal_quantum_counting}
In this section, counting processes that minimize the cost function are considered, starting with the optimal classical counting processes in \cref{subsec:erlang_ladder_clock}, the well-known Erlang ladder.
Then, in \cref{subsec:rank_one_detection_ansatz,subsec:numerical_optimization_discrete_time,subsec:numerical_optimization_continuous_time} we extend the optimization to quantum counting processes.

\subsection{Optimal Classical Process: Erlang Ladder}
\label{subsec:erlang_ladder_clock}
The finite-memory variance bound is tight. 
For every fixed mean $\mudt \geq d$, equality is attained by a discrete-time Erlang-type ladder process. 
To construct it as a quantum model, consider a clock system of $d$ internal states, $\mathcal{H}_d = \operatorname{span} \{ \ket{0}, \ldots, \ket{d-1} \}$, which is initially in the vacuum state $\varrho = \ketbra{0}$. 
At every step, the system advances by one rung with probability $q$ and remains on the same rung with probability $1-q$. An advance from the final state $\ket{d-1}$ produces a tick.

In quantum-instrument notation, the no-tick branch can be represented by
\begin{equation}
    \mathcal{I}_0(\varrho) = A_\mathrm{stay} \varrho A_\mathrm{stay}^{\dagger} + A_\mathrm{adv} \varrho A_{\mathrm{adv}}^{\dagger},
    \label{eq:erlang_notick_map}
\end{equation}
where $A_{\rm stay}=\sqrt{1-q}\,\id$, and
\begin{equation}
    A_\mathrm{adv} = \sqrt{q} \sum_{j=0}^{d-2} \ketbra{j+1}{j}.
    \label{eq:erlang_notick_kraus}
\end{equation}
The tick transition may be written as $A_\mathrm{tick} = \sqrt{q} \, \ketbra{0}{d-1}$, yielding $M_1 = A_\mathrm{tick}^{\dagger} A_\mathrm{tick} = q \ketbra{d-1}$.

Although the tick effect is rank one, the no-tick operation contains two Kraus operators and is therefore intrinsically stochastic as long as $q<1$. 
To attain a mean tick time strictly exceeding $d$, it is necessary for $q$ to be smaller than unity. 
In contrast, for the optimal quantum process, we later consider a coherent, single-Kraus quantum process that surpasses the classical limit. 

The clock must realize $d$ successful advances from the total number of $L$ before ticking. 
The Erlang first-tick distribution is consequently a negative binomial,
\begin{equation}
    p_L = \binom{L-1}{d-1} q^d (1-q)^{L-d}, 
    \qquad L = d, d + 1, \ldots.
    \label{eq:erlang_waiting_distribution}
\end{equation}
with the corresponding mean and variance being
\begin{equation}
    \mudt = \frac{d}{q}, 
    \qquad \sigmadt^2 = \frac{d(1-q)}{q^2}.
    \label{eq:erlang_mean_variance}
\end{equation}
For a prescribed mean $\mudt \geq d$, the transition probability is therefore fixed to $q_\mathrm{Erlang} = {d}/{\mudt}$.
Substitution into \cref{eq:erlang_mean_variance} yields  $\sigmadt^2 = {\mudt(\mudt-d)}/{d}$, and hence
\begin{equation}
    \mathrm{CF}_d = 0.
\end{equation}
The limiting case $q = 1$, or equivalently $\mudt = d$, is the deterministic ladder that advances at every step and ticks exactly at $L = d$.

\subsection{Optimal Quantum Process: Coherent Rank-One Ansatz}
\label{subsec:rank_one_detection_ansatz}
To find the optimal quantum counting process, we consider only pure initial states of the form $\varrho=\ketbra{\psi}$.
The restriction to a pure initial state is without loss for the variance minimization: a mixed initial state produces a convex mixture of waiting-time distributions, and the additional classical mixing cannot reduce their total variance.

Furthermore, we consider a no-tick operation containing only a single Kraus operator,
\begin{equation}
    \mathcal I_0(\varrho)=K_0\varrho K_0^\dagger ,
    \label{eq:single_kraus_notick}
\end{equation}
which in polar form reads
\begin{equation}
    K_0 = U_0 \sqrt{M_0}, 
    \qquad M_0 = K_0^\dagger K_0 \leq \id ,
\end{equation}
with $U_0$ unitary.
We take the complementary tick effect to be rank one,
\begin{equation}
    M_1 = \id-M_0 = (1-q) \ketbra{\Psi}, 
    \qquad 0 \leq q \leq 1.
    \label{eq:rank_one_tick_effect}
\end{equation}
Here \(\ket{\Psi}\) specifies the detectable direction, while \(1-q\) sets the strength of each detection attempt.

The ansatz thus has a simple interpretation. The unitary \(U_0\) coherently moves the conditioned state through Hilbert space, whereas the rank-one detector probes its overlap with a single direction. 
The optimization searches for a trajectory that remains nearly orthogonal to \(\ket{\Psi}\) at early times and approaches it only within a narrow interval around the typical ticking time. 
In terms of the hazard, the desired behavior is an initially suppressed \(h_L\) followed by a rapid increase. 

While choosing an initially pure state does not restrict the minimal achievable tick variance, the rank-one, single-Kraus ansatz is not proven to be the analytically optimal one among all quantum instruments. 
However, optimization over general effects in the small-dimensional cases consistently converges to this structure, and it provides the best solutions found in all dimensions investigated here (\cref{app:discrete_numerical_ansatz}).

\subsection{Discrete-Time Optimization}
\label{subsec:numerical_optimization_discrete_time}
For fixed Hilbert-space dimension \(d\) and prescribed mean first-tick time \(\bar\mu\), we minimize the cost function $\mathrm{CF}_d$ over the initial state and the no-tick dynamics.
Explicitly, the reduced optimization problem is
\begin{equation}
    \begin{aligned}
        \underset{\ket{\psi},\,U_0,\,q,\,\ket{\Psi}}{\operatorname{minimize}} \quad & \mathrm{CF}_d \\
        \text{subject to} \quad & \mudt=\bar\mu, \\
        & 0 \leq q \leq1.
    \end{aligned}
    \label{eq:reduced_fixed_mu_problem_main}
\end{equation}
The moments are evaluated through the exact resolvent formulas derived in \cref{app:computing_mu_sigma_dt}.
The model, defined by \(\ket{\psi}\), \(\ket{\Psi}\), \(U_0\) and $q$ can be fully characterized using $4(d-1)+1$ real parameters, before making use of any additional symmetries of a particular optimum (\cref{app:discrete_numerical_ansatz}).

For each pair \((d,\bar\mu)\), we numerically minimize \cref{eq:reduced_fixed_mu_problem_main} from multiple initial configurations and retain the best converged solution. Repeating the procedure over \(\bar\mu\) gives the optimized curve
\begin{equation}
    \mathrm{CF}^{\mathrm{opt}}_d(\bar\mu) := \min_{\mudt=\bar\mu}\mathrm{CF}_d .
    \label{eq:optimized_CF_curve}
\end{equation}
The resulting parameters are subsequently inserted back into the explicit no-tick dynamics to reconstruct \(p_L\), \(f_L\), and \(h_L\).

Figure~\ref{fig:optimized_distribution_hazard} displays the temporal structure of a representative optimum. 
The first-tick distribution in panel (a) is sharply localized compared with a memoryless geometric process having the same mean. 
This localization can be understood directly from the hazard in panel (c). 
During the initial part of the trajectory, the state remains in the approximately dark subspace orthogonal to \(\ket{\Psi}\), and the probability of ticking at each step is correspondingly small. 
Coherent evolution subsequently drives the conditioned state towards the detector direction, producing a rapid increase of \(h_L\) and concentrating the first-tick probability within a narrow temporal window.

The optimized detector becomes progressively weaker as the prescribed mean is increased. 
Numerically, its strength follows approximately
\begin{equation}
    1-q \propto \bar\mu^{-1}
    \label{eq:detector_strength_scaling}
\end{equation}
in the large-\(\bar\mu\) regime. 
This behavior is consistent with the emergence of a continuous-time limit: the probability of detection in an individual time step vanishes, while the accumulated detection probability over a finite physical time remains nonzero (\cref{app:disc_opt_large_mean}).

For dimension $d=2$, the cost function can be expressed in closed form and the optimal process can be characterized analytically. 
Both mean and variance of the tick time can be expressed as the solution of (nested) linear equations of $2\times 2$ matrices. 
In the limit of very large mean tick time $\mudt\to\infty$, we can express the cost using only rational and trigonometric functions, and we find that (cf.~\cref{subsec_analytical_d2})
\begin{align}
    \lim_{\mudt \to \infty} \left( \frac{\mathrm{CF}_2}{\mudt^2} \right) = -0.5982,
\end{align}
in agreement with the numerically optimized cost function (cf.~\cref{fig:optimized_cost_accuracy}).

\begin{figure}[t]
    \centering
    \includegraphics[width=\linewidth]{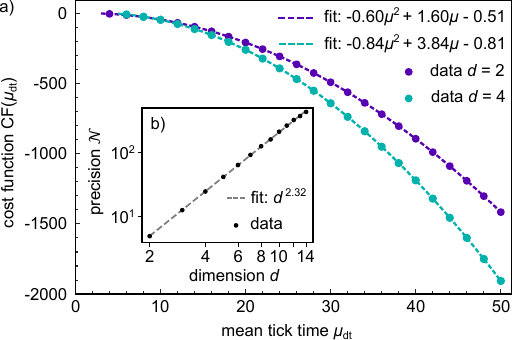}
    \caption{Optimized cost and continuous-time accuracy. 
    (a) Minimum discrete-time cost \(\mathrm{CF}^{\mathrm{opt}}_d(\mudt)\) as a function of the prescribed mean first-tick time for different Hilbert-space dimensions \(d\). 
    Negative values violate the classical $d$-state variance bound. 
    At large \(\mudt\), the optimized curves approach the continuous-time regime. 
    (b) Relative precision \(\mathcal N=\muct^2/\sigmact^2\) of the optimized continuous-time processes as a function of \(d\), on logarithmic axes. 
    The dashed line is a finite-size power-law fit \(\mathcal N_{\mathrm{opt}}\propto d^\alpha\), with \(\alpha\simeq 7/3\). 
    Because the general asymptotic bound of Ref.~\cite{Yang2020} is quadratic in \(d\), this exponent should be understood as an effective scaling over the numerically accessible range and cannot persist at asymptotically large dimension.}
    \label{fig:optimized_cost_accuracy}
\end{figure}

The numerically optimized cost functions are shown in Fig.~\ref{fig:optimized_cost_accuracy}(a). 
For every dimension considered, the optimization produces a range in which \(\mathrm{CF}^{\mathrm{opt}}_d<0\), and hence a waiting-time variance smaller than that of any classical \(d\)-state process with the same mean. 
The quantum improvement becomes more pronounced as the dimension increases. 
Moreover, the large-\(\mudt\) part of each curve approaches a well-defined scaling regime. 
In this regime the per-step detection becomes weak, following the scaling as in \cref{eq:detector_strength_scaling}, and the optimized discrete dynamics converges to a continuous-time no-tick evolution.

\subsection{Continuous-Time Optimization}
\label{subsec:numerical_optimization_continuous_time}
As we explained earlier the direct continuous-time limit analogous of \(\mathrm{CF}_d\) is given by the relative precision $\mathcal N$ (cf.~\cref{eq:N_precision_Def}) and is invariant under a uniform rescaling of all rates. 
It is therefore the natural objective function for the continuous-time optimization. 
As shown in \cref{subsec:discrete_FOM_hazard_bounds}, fixing an arbitrary value of \(\muct\) and minimizing the variance is equivalent to maximizing \(\mathcal N\).

The continuous-time counterpart of the rank-one ansatz is described by the conditional no-tick evolution
\begin{equation}
    \dot\varrho^0 = -i[H,\varrho^0] - \frac{1}{2}
    \left\{
        L^\dagger L,\varrho^0
    \right\},
    \label{eq:dot_rho_simplified}
\end{equation}
with a single rank-one jump effect $L^\dagger = \Gamma\ketbra{\Phi}$. 
For a pure initial state, the unnormalized conditioned wave function obeys
\begin{equation}
    \partial_t |\widetilde \psi_t \rangle = -iH_{\mathrm{eff}} |\widetilde \psi_t \rangle,
    \qquad H_{\mathrm{eff}} = H - \frac{i\Gamma}{2 }\ketbra{\Phi},
    \label{eq:partial_t_tilde_psi}
\end{equation}
with survival probability $f(t) = |\langle \widetilde \psi_t |\widetilde \psi_t \rangle|^2$.

The parameter \(\Gamma\) fixes only the overall time scale. 
We therefore set \(\Gamma = 1\) and maximize \(\mathcal N\) over the dimensionless Hamiltonian \(H/\Gamma\), the initial state \(\ket{\psi}\), and the jump direction \(\ket{\Phi}\). 
The continuous-time search contains \(4(d-1)\) real parameters, one less than the discrete-time problem due to time-scale invariance (\cref{appendix:parametrization_ct_model}). 
The moments \(\muct\) and \(\sigmact^2\) are evaluated directly from the continuous-time Lyapunov equations described in \cref{app:ct_lyapunov}.

The connection with the discrete optimization follows by setting
\begin{equation}
    K_0^{(\delta)} = e^{-iH_{\mathrm{eff}}\delta}.
    \label{eq:ct_discretized_K0}
\end{equation}
For small \(\delta\), we have
\begin{equation}
    K_0^{(\delta)} = e^{-iH\delta} \sqrt{\id - \Gamma \delta \ketbra{\Phi}} + O(\delta^2).
    \label{eq:ct_discretized_polar}
\end{equation}
This precisely matches the expression from \cref{subsec:rank_one_detection_ansatz} with a weak rank-one detector, $(1-q)\ketbra{\Psi}{\Psi}=\Gamma\delta\ketbra{\Phi}{\Phi}/2$, followed by coherent evolution with $U_0=e^{-iH\delta}$. 
The moments then satisfy $\mudt = \muct/\delta + O(1)$ and $\sigmadt^2 = \sigmact^2/\delta^2 + O(\delta^{-1})$ as in \cref{eq:mudt_to_muct,eq:sigmadt_to_sigmact}. 
Thus, the large-\(\mudt\) part of the discrete-time optimization and the direct continuous-time maximization of \(\mathcal N\) probe the same limiting family of processes.

Figure~\ref{fig:optimized_cost_accuracy}(b) shows the resulting dimension dependence. 
Over the dimensions accessible to our optimization, a log--log fit gives
\begin{equation}
    \mathcal N_{\mathrm{opt}}(d) \propto d^\alpha, 
    \qquad \alpha \simeq 2 + \frac{1}{3}.
    \label{eq:finite_size_accuracy_scaling}
\end{equation}
This effective exponent is larger than the approximately quadratic behavior suggested by previously constructed families of quantum clocks~\cite{Woods2022}. 
It suggests a transient, finite-dimensional correction to the asymptotic quadratic bound by Yang and Renner~\cite{Yang2020}.

Taken together, the two optimizations reveal the same physical mechanism. 
The optimal process combines a long approximately dark evolution with a coherently generated, sharply localized detection window. 
In discrete time, this structure minimizes \(\mathrm{CF}_d\) at fixed \(\mudt\); in the weak-detection limit it becomes a continuous-time rank-one jump process maximizing the scale-independent accuracy \(\mathcal N\).

\section{Discussion and outlook}
\label{sec:discussion_outlook}
We investigated the role of internal memory as a resource for classical and quantum counting processes. 
We first established a tight finite-memory variance bound for arbitrary classical $d$-state counting processes, thereby resolving the conjecture of Ref.~\cite{Budroni2021Ticking-clock}. 
For mean tick times larger than the memory dimension, the bound is saturated by the discrete Erlang ladder, which therefore provides a sharp classical benchmark for temporal precision. 
The classical bound can equivalently be viewed as a discrete-time trade-off between memory (dimension), resolution (inverse mean tick time), and precision (variance), complementing the precision--resolution~\cite{Meier2023} and precision--memory constraints~\cite{Woods2019,Woods2022,Yang2020} established for continuous-time clocks. 
Violations of the corresponding cost-function bound consequently certify first-tick statistics that cannot be reproduced by any classical process with the same memory. 

Our subsequent numerical optimization then addresses the problem of identifying finite-dimensional quantum counting processes that maximally violate this classical benchmark, demonstrating, at fixed memory dimension, a more peaked first-tick distribution as compared to a classical Erlang process. 
The best processes found consistently exhibit a simple structure, with a pure initial state, single-Kraus no-tick evolution, and a rank-one tick effect. 
For $d=2$, this numerical picture is complemented by an analytical treatment of the large-mean regime within this single-Kraus family, showing that the quantum advantage persists asymptotically in this class.

Intuitively, from the point of view of the hazard sequence, an extremal conditioned dynamics keeps the system weakly detectable at early times and drives it toward the detectable direction only around the typical ticking time, thereby concentrating the first-tick distribution. 
In the classical Erlang mechanism, this is achieved through stochastic progression across distinguishable states rather than coherent conditioned evolution. 
In contrast, in the quantum case, the numerical results point to coherent single-Kraus dynamics with rank-one detection as the relevant ingredients for optimality.
A rigorous proof of the optimality of this ansatz remains an open question.

In the large-mean regime, the optimized discrete-time processes enter a weak-detection limit and approach the same rank-one quantum-jump family obtained from the direct continuous-time optimization. 
This connects the discrete-time variance problem with the scale-independent continuous-time precision problem. 
Over the dimensions accessible numerically, the best continuous-time processes found exhibit an effective scaling close to $d^{7/3}$, which is faster than the quadratic scaling from~\cite{Woods2022,Yang2020,Dost2023}. 
An important open problem is to determine the crossover to the asymptotic regime and to establish whether the optimized family eventually approaches the quadratic upper bound, and with which prefactor.

Several questions concerning quantum optimality remain open. 
In particular, it would be interesting to determine whether the single-Kraus, rank-one family identified numerically is globally optimal, and to obtain analytical bounds beyond the qubit case. 
More generally, a characterization of the optimal quantum precision at arbitrary memory dimension would provide a natural quantum counterpart to the tight classical benchmark established here, and characterize the largest departure from classicality in this context, with applications
reaching also the realm of quantum foundations~\cite{Budroni2021Ticking-clock,Vitagliano2023}. 
Beyond quantum timekeeping, the continuous-time optimization is closely related to the problem of minimizing the coefficient of variation within classes of matrix-exponential distributions~\cite{Bladt2017}. 
Progress on the quantum optimization problem may therefore also provide insight into identifying optimal matrix-exponential distributions~\cite{Meszaros2022,Horvath2020,Battagliola2026}, which thus motivates an even broader scope investigation of the questions addressed in this work.

\acknowledgments
The authors thank Ali Asadian, Mohammad Mehboudi,  Nuriya Nurgalieva, Ralph Silva, Lucas Vieira, and Yuxiang Yang for discussions. 
BO acknowledges the Ernst Mach Grant - worldwide, administered by OeAD and financed by the Federal Ministry Women, Science and Research Republic of Austria (BMFWF). 
FM is co-funded by the European Union through the ERC Consolidator Grant Cocoquest (Grant Agreement No.~101043705), and through the ERC Synergy Grant SuperWave (Grant Agreement No.~101071882) and by the Austrian Science Fund FWF (Grant DOI:~\href{https://doi.org/10.55776/COE1}{10.55776/COE1}). 
This research was funded in whole or in part by the Austrian Science Fund FWF (Grant DOI:~\href{https://doi.org/10.55776/P35810}{10.55776/P35810}, and~\href{https://doi.org/10.55776/P36633}{10.55776/P36633}). 
BO, PB and GV also acknowledge support from the Grant No.~RYC2024-048278-I funded by MCIU/AEI (Grant DOI:~\href{https://doi.org/10.13039/501100011033}{10.13039/501100011033}) and FSE+.

\textit{AI Disclosure Statement.---}%
The authors disclose the use of LLM tools (ChatGPT, Gemini) for assistance in mathematical derivations, writing, and coding.
All the text is written by the authors, and they take full responsibility for its contents.

\textit{Author Contributions.---}%
BO analyzed the discrete-time model, and FM the continuous-time model.
CB proved the classical lower bound.
PB and GV supervised the project.
All authors contributed equally to verifying the correctness of the technical results and to writing and revising the manuscript.

Views and opinions expressed are, however, those of the authors only and do not necessarily reflect those of the European Union, and the European Union can not be held responsible for them.

\section*{Data Availability}
The data and numerical code used in this work are available upon reasonable request.

\section*{Appendices}
\appendix
\crefalias{section}{appendix}
\crefalias{subsection}{appendix}

\section{Classical Finite-Memory Bound}
\label{app:classical_bound}
In this appendix, we prove the classical finite-memory variance bound introduced in \cref{subsec:classical_variance_benchmark}. 
The proof is based on an exact correspondence between the discrete-time first-tick process and a continuous-time absorbing Markov process with the same number of transient internal states. 
This allows us to apply the continuous-time variance bound of Aldous and Shepp~\cite{AldousShepp1987} and transfer it back to the discrete-time setting.

\subsection{Classical First-Tick Process as an Absorption Problem}
\label{app:classical_moments}
As discussed in \cref{sbsec:dt-count-process}, throughout this work we assume that after each tick the clock is reset to the same reference initial state. 
The successive inter-tick waiting times $\tau_1, \tau_2, \ldots$ are therefore independent and identically distributed. 
Consequently, for the purpose of bounding their mean and variance it is sufficient to consider a single cycle, which we take to be the first waiting time $\tau_1$.

We recall the classical specialization of the counting process introduced in the main text. 
A classical clock with $d$ internal states is specified by an initial probability row vector
\begin{equation}
    \pi \in \mathbb R^d, 
    \qquad \pi_i \geq 0, 
    \qquad \pi \eta = 1,
\end{equation}
where
\begin{equation}
    \eta = (1, \ldots, 1)^{\mathsf T},
\end{equation}
and by two nonnegative transition matrices $T_0$ and $T_1$ associated, respectively, with the no-tick and tick outcomes. 
Their normalization condition reads
\begin{equation}
    T_0 + T_1 \geq 0, 
    \qquad (T_0 + T_1) \eta = \eta.
    \label{eq:classical_normalization_appendix}
\end{equation}
In particular, the no-tick matrix $T_0$ is substochastic,
\begin{equation}
    T_0 \eta \leq \eta.
    \label{eq:classical_substochastic_appendix}
\end{equation}
The row deficit of $T_0$ is the probability of producing a tick at the next step. 
We denote the corresponding column vector by
\begin{equation}
    r := T_1 \eta = (\id - T_0) \eta.
    \label{eq:classical_absorption_vector}
\end{equation}

The first-tick distribution introduced in \cref{eq:p(L)_def} therefore reduces, in the classical case, to
\begin{equation}
    \mathrm{Prob}(\tau_1 = L) = p_L 
    = \pi T_0^{L-1} T_1 \eta 
    = \pi T_0^{L-1} (\id - T_0) \eta, 
    \label{eq:classical_first_tick}
\end{equation}
where $L \geq 1$.
For first-tick statistics, the internal state reached after the tick is irrelevant. 
We may therefore replace all tick transitions by a transition to a single absorbing state $\partial$. 
Equivalently, one can consider the augmented stochastic transition matrix
\begin{equation}
    \widehat T =
    \begin{pmatrix}
        T_0 & r\\
        0   & 1
    \end{pmatrix}.
    \label{eq:classical_absorbing_transition_matrix}
\end{equation}
With this representation, $\tau_1$ is precisely the absorption time into $\partial$. 
Thus, the reset clock and the absorbing first-passage process are two equivalent descriptions of one inter-tick cycle. 
The reset mechanism is needed to generate a renewal process over many ticks, whereas the variance bound below depends only on this single-cycle absorption problem. 
It is clear from Eq.~\eqref{eq:classical_absorbing_transition_matrix} that there is a one-to-one mapping between absorption processes and counting processes where
the ticking and absorption probabilities are mapped onto one another.

We restrict the attention to the sector of internal states that is reachable from the support of $\pi$ before the first tick. 
States that are never reached during a cycle do not contribute to the first-tick distribution and may be discarded. 
Physically, we consider proper ticking clocks, for which the tick occurs with probability one from the relevant sector. 
We therefore impose the transience condition
\begin{equation}
    \varrho (T_0) < 1,
    \label{eq:classical_transience_appendix}
\end{equation}
where $\varrho(T_0)$ denotes the spectral radius of $T_0$, on this reachable sector, i.e., on the states reached by some transition. 
In finite dimension this condition guarantees that the probability of indefinite survival vanishes and that all moments of the first-tick time are finite. 
Indeed, if the reachable sector contains fewer than $d$ states, the argument below can be applied to that smaller sector and gives a bound at least as strong as the $d$-state bound.

Let us now recall the statement of \cref{thm:classical_variance_bound} from the main text, which is that any classical discrete-time counting process with at most $d$
reachable internal states, arbitrary initial distribution $\pi$, and transient no-tick matrix $T_0$ satisfies
\begin{equation}
    d\,\sigmadt^2 \geq \mudt (\mudt - d),
    \label{eq:classical_variance_bound_appendix}
\end{equation}
or equivalently,
\begin{equation}
    \sigmadt^2 + \mudt \geq \frac{\mudt^2}{d}.
    \label{eq:classical_variance_bound_shifted}
\end{equation}

In the following, we will proceed with a proof of this statement, considering the shifted form in \cref{eq:classical_variance_bound_shifted}. 
First, in the next subsection, we construct from $T_0$ an auxiliary continuous-time absorbing Markov process whose absorption time has variance proportional to $\sigmadt^2+\mudt$, while preserving the number of transient internal states.

\subsection{Exact Poissonization of the Discrete First-Tick Process}
\label{app:classical_poissonization}
We now associate with the discrete-time first-tick process of \cref{app:classical_moments} an auxiliary continuous-time absorbing Markov process. 
The construction is an exact Poissonization (or uniformization) of the discrete dynamics~\cite{Serfozo2009}. 
Crucially, it does not increase the number of transient internal states.

The Poissonization maps the discrete-time Markov process to a continuous-time one as follows:
An auxiliary Poisson clock of rate $\lambda > 0$ is considered, and we let $E_1, E_2, \ldots$ denote the successive waiting times of such a clock.
After each auxiliary jump, the state then changes according to the original discrete-time update.
Moreover, if the first tick of the discrete-time process occurs after $\tau_1$ updates (jumps), we define the corresponding auxiliary continuous absorption time as 
\begin{equation} 
    \widetilde{\tau} := \sum_{k=1}^{\tau_1} E_k. 
    \label{eq:poissonized_absorption_time} 
\end{equation}
Note that both $\tau_1$ (number of jumps) and the $E_k$'s (successive jump waiting times) are random variables.

We take the $E_k$'s to be independent and identically distributed (i.i.d.) according to an exponential distribution of rate $\lambda$, which we write as
\begin{equation}
    E_1, E_2, \ldots \overset{\mathrm{iid}}{\sim} \operatorname{Exp}(\lambda),
    \qquad \lambda > 0.
    \label{eq:poisson_waiting_times}
\end{equation}
The symbol $\sim$ means ``is distributed according to''.
Explicitly, each $E_k$ is a nonnegative random variable with survival probability
\begin{equation}
    \operatorname{Prob}(E_k > t) 
    = e^{-\lambda t}, 
    \qquad t \geq 0.
    \label{eq:exponential_survival_probability}
\end{equation}
In particular, mean and variance are given by
\begin{equation}
    \mathbb E[E_k] = \frac{1}{\lambda}, 
    \qquad \operatorname{Var}(E_k) = \frac{1}{\lambda^2}.
    \label{eq:exponential_waiting_time_moments}
\end{equation}

Conditioning on the event $\tau_1=L$ fixes the number of exponential inter-arrival times appearing in \cref{eq:poissonized_absorption_time}. 
Thus
\begin{equation}
    \widetilde{\tau} \mid(\tau_1 = L) = E_1 + \cdots + E_L.
    \label{eq:conditional_absorption_sum}
\end{equation}
The notation on the left denotes the random variable $\widetilde{\tau}$ conditioned on the event that the discrete first-tick time is equal to $L$. 
A sum of $L$ independent exponential random variables with common rate $\lambda$ is said to have an Erlang distribution with shape parameter $L$ and rate $\lambda$~\cite{Erlang1917,Serfozo2009}. 
We denote this by
\begin{equation}
    \widetilde{\tau} \mid(\tau_1 = L) \sim \operatorname{Erlang}(L,\lambda).
    \label{eq:conditional_erlang}
\end{equation}
Explicitly, the corresponding probability density is
\begin{equation}
    {\rm Prob} (\widetilde{\tau} = t | \tau_1 = L) 
    = \frac{\lambda^L t^{L-1} e^{-\lambda t}}{(L-1)!}, 
    \qquad t \geq 0.
    \label{eq:conditional_erlang_density}
\end{equation}
Its mean and variance are~\cite{Serfozo2009},
\begin{equation}
    \mathbb E[\widetilde{\tau} \mid\tau_1 = L] = \frac{L}{\lambda},
    \qquad \operatorname{Var}(\widetilde{\tau} \mid\tau_1 = L) = \frac{L}{\lambda^2}.
    \label{eq:conditional_erlang_moments}
\end{equation}

The unconditional distribution of $\widetilde{\tau}$ is a mixture over the possible values of $\tau_1$ and, in general, need not itself be an Erlang distribution.

We next identify the continuous-time generator associated with this construction. 
During an infinitesimal interval $\dd t$, the auxiliary Poisson clock jumps with probability
\begin{equation}
    \operatorname{Prob} (\text{one jump in }dt) = \lambda\,\dd t + o(\dd t),
    \label{eq:poisson_ring_short_time}
\end{equation}
while the probability of more than one jump is $o(\dd t)$. Conditional on a jump, the transient state is updated according to $T_0$. 
Therefore, the transient evolution over the interval $dt$ is
\begin{equation}
    (1 - \lambda \,dt) \id + \lambda\,dt\,T_0 + o(dt) = \id + \lambda(T_0-\id)\,dt + o(dt).
    \label{eq:poissonized_short_time_evolution}
\end{equation}
Comparing this expression with the standard short-time expansion
\begin{equation}
    P(dt) = \id + Q_\lambda \,dt + o(dt),
    \label{eq:generator_short_time_expansion}
\end{equation}
we obtain the transient generator
\begin{equation}
    Q_\lambda = \lambda(T_0-\id).
    \label{eq:poissonized_generator}
\end{equation}

The absorption-rate vector is determined by the row deficit of $Q_\lambda$. 
Using \cref{eq:classical_absorption_vector},
\begin{equation}
    - Q_\lambda \eta = \lambda (\id- T_0) \eta = \lambda r.
    \label{eq:poissonized_absorption_rates}
\end{equation}
Hence the full continuous-time generator, including the absorbing tick state $\partial$, is
\begin{equation}
    \widehat Q_\lambda =
    \begin{pmatrix}
        \lambda(T_0-\id) & \lambda r\\
        0                & 0
    \end{pmatrix}.
    \label{eq:poissonized_full_generator}
\end{equation}
For $i\neq j$,
\begin{equation}
    (Q_\lambda)_{ij} = \lambda(T_0)_{ij} \geq 0,
    \label{eq:poissonized_offdiagonal_rates}
\end{equation}
while
\begin{equation}
    (Q_\lambda)_{ii} = \lambda \bigl((T_0)_{ii} - 1\bigr) \leq 0.
    \label{eq:poissonized_diagonal_rates}
\end{equation}
The resulting process is therefore a valid continuous-time absorbing Markov process. 
Moreover, since $\varrho(T_0) < 1$, all eigenvalues of $Q_\lambda = \lambda(T_0 - \id)$ have strictly negative real part on the reachable sector, and the continuous-time process is transient.

It is useful to clarify the role of the diagonal entries of $T_0$ in this construction. 
A jump of the auxiliary Poisson clock may result in a transition $i\to i$. Such an event contributes one discrete update to $\tau_1$, even though the visible internal state does not change. 
In the continuous-time representation these events are usually referred to as virtual transitions. 
The rate of an actual departure from the transient state $i$ is therefore
\begin{equation}
    - (Q_\lambda)_{ii} = \lambda \bigl(1 - (T_0)_{ii}\bigr),
    \label{eq:poissonized_escape_rate}
\end{equation}
while the absorption component of this rate is $\lambda r_i$.

A central feature of the construction is that the auxiliary timing mechanism introduces no additional Markov memory, and that exponential waiting times are memoryless. 
The Poissonized process therefore has exactly the same transient internal states as the reachable sector of the original discrete process, which ensures the applicability of the Aldous--Shepp bound~\cite{AldousShepp1987} in the following steps.

In the terminology of phase-type distributions, these transient states are called \emph{phases}. 
Thus, a process with $d$ transient internal states is also referred to as a $d$-phase absorbing process; the absorbing state $\partial$ itself is not counted as a phase. 
We will use this terminology only when referring to the continuous-time absorption-time results below.

Finally, we emphasize that the Poissonization above is an exact stochastic embedding of the fixed discrete-time process, rather than a continuous-time approximation. 
Indeed, from \cref{eq:poissonized_generator},
\begin{equation}
    e^{Q_\lambda t} = e^{\lambda (T_0 - \id) t} = e^{-\lambda t} e^{\lambda t T_0} 
    = e^{-\lambda t} \sum_{n=0}^{\infty} \frac{(\lambda t)^n}{n!} T_0^n.
    \label{eq:poissonized_semigroup}
\end{equation}
If $N_t$ denotes the number of jumps of a rate-$\lambda$ Poisson process up to time $t$, then
\begin{equation}
    \operatorname{Prob}(N_t = n) = e^{-\lambda t} \frac{(\lambda t)^n}{n!},
    \label{eq:poisson_number_rings}
\end{equation}
and therefore
\begin{equation}
    e^{Q_\lambda t} = \sum_{n=0}^{\infty} \operatorname{Prob}(N_t = n)\, T_0^n.
    \label{eq:poissonized_semigroup_mixture}
\end{equation}
Thus, evolving the auxiliary continuous-time process for a duration $t$ is exactly equivalent to drawing the number $N_t$ of Poisson jumps and applying the original discrete no-tick transition matrix $T_0$ exactly $N_t$ times.
In the next subsection, we relate the first two moments of its absorption time $\widetilde{\tau}$ to the discrete first-tick moments $\mudt$ and $\sigmadt^2$.

\subsection{Relation Between Discrete and Continuous First-Tick Moments}
\label{app:classical_moment_bridge}
We now relate the first two moments of the discrete first-tick time $\tau_1$ to those of the auxiliary continuous absorption time $\widetilde{\tau}$.
We can rewrite Eq.~\eqref{eq:conditional_erlang_moments} directly in terms of the variable $\tau_1$ as 
\begin{equation}
    \mathbb E\!\left[ \widetilde{\tau} \mid \tau_1 \right] = \frac{\tau_1}{\lambda},
    \qquad \operatorname{Var}\! \left( \widetilde{\tau} \mid \tau_1 \right) = \frac{\tau_1}{\lambda^2}.
    \label{eq:conditional_poissonized_moments}
\end{equation}

The mean absorption time follows immediately from the law of total expectation (see e.g., \cite{BertsekasTsitsiklis2008}),
\begin{equation}
    \mathbb E[\widetilde{\tau}] = \mathbb E\!\left[ \mathbb E[ \widetilde{\tau}\mid\tau_1 ] \right].
    \label{eq:total_expectation_poissonized}
\end{equation}
Using \cref{eq:conditional_poissonized_moments} and $\mathbb E[\tau_1] = \mudt$, we obtain
\begin{equation}
    \mathbb E[\widetilde{\tau}] = \frac{\mathbb E[\tau_1]}{\lambda} = \frac{\mudt}{\lambda}.
    \label{eq:poissonized_mean_bridge}
\end{equation}
The law of total variance (see e.g., \cite{BertsekasTsitsiklis2008}) gives
\begin{equation}
    \operatorname{Var}(\widetilde{\tau}) 
    = \mathbb E\!\left[ \operatorname{Var}\!\left( \widetilde{\tau} \mid \tau_1 \right) \right]
    + \operatorname{Var}\!\left( \mathbb E[ \widetilde{\tau} \mid \tau_1 ] \right).
    \label{eq:total_variance_poissonized}
\end{equation}
The first term accounts for the fluctuations of the exponential waiting times at fixed discrete first-tick time. 
Using \cref{eq:conditional_poissonized_moments},
\begin{equation}
    \mathbb E\!\left[ \operatorname{Var}\!\left( \widetilde{\tau} \mid \tau_1 \right) \right] 
    = \frac{\mathbb E[\tau_1]}{\lambda^2} 
    = \frac{\mudt}{\lambda^2}.
    \label{eq:poissonized_conditional_variance_term}
\end{equation}
The second term accounts for the fluctuations of the discrete first-tick time itself,
\begin{equation}
    \operatorname{Var}\!\left( \mathbb E[ \widetilde{\tau} \mid \tau_1 ] \right)
    = \operatorname{Var}\!\left( \frac{\tau_1}{\lambda} \right)
    = \frac{\sigmadt^2}{\lambda^2}.
    \label{eq:poissonized_variance_of_conditional_mean}
\end{equation}
Combining the two contributions gives
\begin{equation}
    \operatorname{Var}(\widetilde{\tau}) = \frac{\sigmadt^2+\mudt}{\lambda^2}.
    \label{eq:poissonized_variance_bridge}
\end{equation}

We thus obtain the exact moment correspondence
\begin{equation}
    \mathbb E[\widetilde{\tau}] = \frac{\mudt}{\lambda}, 
    \qquad \operatorname{Var}(\widetilde{\tau}) = \frac{\sigmadt^2 + \mudt}{\lambda^2}.
    \label{eq:poissonized_moment_bridge}
\end{equation}
By \cref{app:classical_poissonization}, this auxiliary process has the same number of transient phases as the reachable sector of the original discrete-time clock. 
We can therefore combine \cref{eq:poissonized_moment_bridge} with the continuous-time variance bound for finite phase-type processes, which we do in the next subsection.

\subsection{Application of the Aldous--Shepp Variance Bound}
\label{app:aldous_shepp_application}
The final ingredient is the continuous-time variance bound of Aldous and Shepp~\cite{AldousShepp1987} for absorption times of finite-state continuous-time Markov processes. We recall the statement of their result.

\begin{theorem}[Aldous--Shepp variance bound]
\label{thm:aldous_shepp}
    Let $\tau$ be the absorption time of a finite-state continuous-time Markov process that starts from a specified transient state and evolves through at most $d$ transient states before absorption. 
    Assuming that absorption occurs with probability one and that the second moment of $\tau$ is finite,
    \begin{equation}
        \operatorname{Var}(\tau) \geq \frac{\mathbb E[\tau]^2}{d}.
        \label{eq:aldous_shepp_bound}
    \end{equation}
    Equality is attained by a serial chain of $d$ exponential stages with equal rates, whose absorption time follows an Erlang distribution with shape parameter $d$.
\end{theorem}
We note that by the one-to-one mapping between reset clocks and absorbing process, see the construction in \cref{eq:classical_absorbing_transition_matrix}, the dimensional bound on clock precision in Ref.~\cite{Woods2022} can be obtained directly from \cref{thm:aldous_shepp}.

We saw in \cref{app:classical_poissonization} that the auxiliary process generated by $\widehat Q_\lambda$ is a continuous-time absorbing Markov process with at most $d$ transient phases. 
The absorbing state $\partial$ is not counted among these phases. 
Moreover, the transience condition $\varrho(T_0)<1$ guarantees almost-sure absorption on the reachable sector, while \cref{eq:poissonized_moment_bridge} guarantees that the first two moments of the absorption time $\widetilde{\tau}$ are finite. 
Hence \cref{thm:aldous_shepp} applies.

To prove our bound, we first consider a definite initial internal state,
\begin{equation}
    \pi = e_i^{\mathsf T},
    \label{eq:pure_initial_state}
\end{equation}
for some reachable state $i$. 
Applying \cref{eq:aldous_shepp_bound} to the auxiliary absorption time $\widetilde{\tau}$ gives
\begin{equation}
    \operatorname{Var}(\widetilde{\tau}) \geq \frac{\mathbb E[\widetilde{\tau}]^2}{d}.
    \label{eq:aldous_shepp_poissonized}
\end{equation}
Substituting the exact moment relations
\begin{equation}
    \mathbb E[\widetilde{\tau}] = \frac{\mudt}{\lambda},
    \qquad \operatorname{Var}(\widetilde{\tau}) = \frac{\sigmadt^2 + \mudt}{\lambda^2},
\end{equation}
obtained in \cref{eq:poissonized_moment_bridge}, we find
\begin{equation}
    \frac{\sigmadt^2 + \mudt}{\lambda^2} \geq \frac{1}{d} \frac{\mudt^2}{\lambda^2}.
    \label{eq:aldous_shepp_substitution}
\end{equation}
The arbitrary Poisson rate $\lambda$ cancels, yielding
\begin{equation}
    \sigmadt^2 + \mudt \geq \frac{\mudt^2}{d}.
    \label{eq:classical_bound_pure_state_shifted}
\end{equation}
Equivalently,
\begin{equation}
    d\,\sigmadt^2 \geq \mudt(\mudt-d).
    \label{eq:classical_bound_pure_state}
\end{equation}

The cancellation of $\lambda$ shows explicitly that the final discrete-time bound is independent of the rate of the auxiliary Poisson clock. Intuitively, it can be thought as fixing the unit of continuous time.

It remains to extend the result from a definite initial state to an arbitrary initial probability distribution $\pi$. For each reachable internal state $i$, define the conditional first-tick mean and variance
\begin{equation}
    \mu_i := \mathbb E[\tau_1\mid X_0=i],
    \qquad \sigma_i^2 := \operatorname{Var}(\tau_1\mid X_0=i).
    \label{eq:state_conditioned_moments}
\end{equation}
The pure-state result \cref{eq:classical_bound_pure_state_shifted} implies that, for every reachable state $i$,
\begin{equation}
    \sigma_i^2+\mu_i \geq \frac{\mu_i^2}{d}.
    \label{eq:state_conditioned_variance_bound}
\end{equation}

For an arbitrary initial distribution $\pi$, the law of total expectation gives
\begin{equation}
    \mudt = \sum_i\pi_i\mu_i.
    \label{eq:mixture_mean}
\end{equation}
Similarly, the law of total variance~\cite{BertsekasTsitsiklis2008} gives
\begin{equation}
    \sigmadt^2 =
    \sum_i \pi_i \sigma_i^2 + \operatorname{Var}_{\pi}(\mu_i),
    \label{eq:mixture_variance}
\end{equation}
where
\begin{equation}
    \operatorname{Var}_{\pi}(\mu_i) := \sum_i\pi_i\mu_i^2 - \left( \sum_i\pi_i\mu_i \right)^2 \geq 0.
    \label{eq:variance_over_initial_state}
\end{equation}

Combining \cref{eq:mixture_mean,eq:mixture_variance}, we obtain
\begin{equation}
    \begin{split}
        \sigmadt^2 + \mudt &= \sum_i \pi_i (\sigma_i^2 + \mu_i) + \operatorname{Var}_{\pi}(\mu_i) \\
                           &\geq \frac{1}{d} \sum_i \pi_i \mu_i^2 + \operatorname{Var}_{\pi}(\mu_i),
    \end{split}
    \label{eq:mixture_bound_intermediate}
\end{equation}
where the inequality follows from \cref{eq:state_conditioned_variance_bound}. 
Using
\begin{equation}
    \sum_i \pi_i \mu_i^2 = \mudt^2 + \operatorname{Var}_{\pi}(\mu_i),
    \label{eq:mixture_second_moment_identity}
\end{equation}
we arrive at
\begin{equation}
    \begin{split}
        \sigmadt^2 + \mudt &\geq \frac{\mudt^2}{d} + \left( 1 + \frac{1}{d} \right) \operatorname{Var}_{\pi}(\mu_i)\\
        &\geq \frac{\mudt^2}{d}.
    \end{split}
    \label{eq:mixture_final_bound}
\end{equation}
Therefore,
\begin{equation}
    d\,\sigmadt^2 \geq \mudt (\mudt - d),
    \label{eq:classical_variance_bound_general_pi}
\end{equation}
for an arbitrary initial distribution $\pi$. This completes the proof of \cref{thm:classical_variance_bound}.

To show that the bound is tight, it is sufficient to verify that it is saturated by some explicit model.  As shown in \cref{subsec:erlang_ladder_clock}, the discrete Erlang-type ladder with transition parameter $q$ has first-tick moments
\begin{equation}
    \mudt = \frac{d}{q}, 
    \qquad \sigmadt^2 = \frac{d(1-q)}{q^2},
    \label{eq:erlang_ladder_moments_appendix}
\end{equation}
and consequently,
\begin{equation}
    d\,\sigmadt^2 = \mudt (\mudt - d),
    \label{eq:erlang_ladder_saturates_bound}
\end{equation}
The inequality of \cref{thm:classical_variance_bound} is thus saturated by this family.

\section{Derivation of the Hazard Bound}
\label{app:derivation_hazard_1}
Here, we prove \cref{eq:CF_hazard_bound_discrete_rewritten} from the main text. 
To start, let us we recall $f_L = \Pr(T>L)$, the probability of surviving the first \(L\) steps, with \(f_0=1\). 
Using the definition in \cref{eq:dt_hazard_def} of the main text, we can write the survival probability as
\begin{equation}
    f_L = \prod_{j=1}^{L} (1 - h_j).
    \label{eq:survival_from_hazard}
\end{equation}

Moreover, we can rewrite the definition for the mean tick time and the variance as,
\begin{align}
    \mudt &= \sum_{L=0}^{\infty}f_L,
    \label{eq:mean_survival_main} \\
    \sigmadt^2 &= 2\sum_{L=0}^{\infty} (L + 1) f_L - \mudt^2 - \mudt.
    \label{eq:variance_survival_main}
\end{align}
Consequently, the cost can be expressed as
\begin{equation}
    \mathrm{CF}_d = 2d \sum_{L=0}^{\infty} (L + 1) f_L - (d + 1)\mudt^2.
    \label{eq:CF_hazard_direct_rewritten}
\end{equation}
At fixed mean, the only remaining dynamical contribution is therefore the weighted tail \(\sum_L(L+1)f_L\).

To obtain a bound on the cost function, we may define the largest hazard attained along the conditioned trajectory as
\begin{equation}
    h_{\max} := \sup_{L\geq1} h_L.
    \label{eq:hmax_def_rewritten}
\end{equation}
Using $h_L = p_L / f_{L-1}$ from \cref{eq:dt_hazard_def}, and the definition of $\mudt$, we can write
\begin{equation}
    \mudt = \sum_{L=1}^{\infty} L f_{L-1}h_L,
    \label{eq:mu_hazard_rewritten}
\end{equation}
which is upper bounded by
\begin{align}
    \mudt &\leq h_{\max} \sum_{L=1}^{\infty}L f_{L - 1} \nonumber\\
          &= h_{\max} \sum_{L=0}^{\infty} (L + 1) f_L.
\end{align}
It follows that
\begin{equation}
    \sum_{L = 0}^{\infty} (L + 1) f_L \geq \frac{\mudt}{h_{\max}},
    \label{eq:tail_hazard_bound_rewritten}
\end{equation}
and therefore Eq.~\eqref{eq:CF_hazard_bound_discrete_rewritten} from the main text, which we repeat here for convenience:
\begin{equation}
    \mathrm{CF}_d \geq \frac{2d\,\mudt}{h_{\max}} - (d + 1) \mudt^2.
\end{equation}
Figure~\ref{fig:hazard_bound_illustration} illustrates the resulting constraint imposed by the maximal hazard. 
For fixed dimension \(d\) and mean waiting time \(\mudt\), a smaller value of \(h_{\max}\) raises the lower bound in \cref{eq:CF_hazard_bound_discrete_rewritten}.

\begin{figure}[t]
    \centering
    \includegraphics[width=\linewidth]{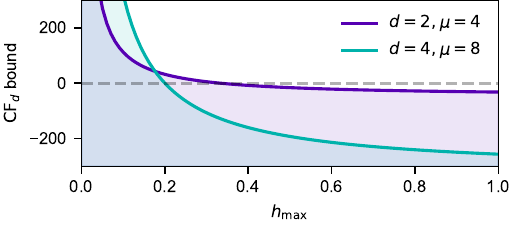}
    \caption{Illustration of the hazard-based lower bound on the discrete-time cost function, \(\mathrm{CF}_d \geq 2d \mudt / h_{\max} - (d + 1) \mudt^2\). 
    At fixed \(d\) and \(\mudt\), a uniformly smaller maximal hazard forces a larger weighted survival tail and therefore a larger cost. 
    The shaded region visualizes the forbidden region implied by the bound.}
    \label{fig:hazard_bound_illustration}
\end{figure}

\section{Computing the Cost Function using Lyapunov Equations}
\subsection{Computing $\mu$ and $\sigma^2$ in Discrete-Time}
\label{app:computing_mu_sigma_dt}
Let us consider the moments of the first tick distribution~\cite{Budroni2021Ticking-clock}.
Define the $Z$-transform of a sequence $g:\mathbb N\to \mathbb{R}$ by
\begin{equation}
    \ZZ[g](z) = \sum_{n=0}^\infty g(n) z^{-n} =: \tilde g(z),
\end{equation}
and the generating function
\begin{equation}
    Q(z) := \ZZ [p_L] (z).
\end{equation}
From \cref{eq:pL_fL} and the standard shift/difference identities of the $Z$-transform, one obtains
\begin{equation}
    Q(z) = -(1 - 1/z) \, \tilde f(z),
    \label{eq:Q_from_f}
\end{equation}
with
\begin{equation}
    \tilde f(z) := \ZZ [f_L] (z).
    \label{eq:tilde_f_def}
\end{equation}
Derivatives of $Q$ at $z=1$ generate the moments of the first-tick distribution. 
In particular,
\begin{equation}
    \begin{aligned}
        \mu &:= \sum_{L\ge1} L\,p_L = -\left. \frac{\partial Q(z)}{\partial z} \right|_{z = 1},\\
        \delta^2 &:= \sum_{L \ge 1} L^2 p_L 
        = \left. \frac{\partial^2 Q(z)}{\partial z^2} \right|_{z = 1} 
        + \left. \frac{\partial Q(z)}{\partial z} \right|_{z = 1},
    \end{aligned}
    \label{eq:moments_from_Q}
\end{equation}
and the variance is $\sigma^2 = \delta^2 - \mu^2$. 
Combining \cref{eq:Q_from_f,eq:moments_from_Q} yields~\cite{Budroni2021Ticking-clock}:
\bea
    \mu &= \tilde f(1), \\
    \sigma^2 &= -\mu (\mu - 1) - 2\tilde f'(1).
    \label{eq:moments_from_f}
\eea

Our figures of merit can be obtained from the resolvent of the CP map $\mathcal I_0$.
From \cref{eq:tilde_f_def} one has
\begin{equation}
    \tilde f(z) = \tr \left[(\idmap - \tfrac1z\mathcal I_0)^{-1}(\varrho)\right].
    \label{eq:ftilde_resolvent}
\end{equation}
Let $X_1(z) := (\idmap - \tfrac1z\mathcal I_0)^{-1}(\varrho)$. 
Then $X_1(z)$ is solution of the discrete-time Lyapunov equation
\begin{equation}
    \varrho = X_1(z) - \frac{1}{z} \mathcal{I}_0 (X_1(z)) .
    \label{eq:Lyap1}
\end{equation}
In particular, evaluating at $z=1$ gives
\begin{equation}
    \mu = \tilde f(1) = \tr \big[X_1(1)\big].
    \label{eq:mu_Lyap}
\end{equation}

To compute $\tilde f'(1)$ (hence $\sigma^2$) without explicitly differentiating $X_1(z)$, we use the identity
\begin{equation}
    \tilde f'(1) = \mu - \tr(X_2),
\end{equation}
with
\begin{equation}
    X_2 := (\idmap - \mathcal I_0)^{-2}(\varrho),
\end{equation}
where $X_2$ can be obtained \emph{by nesting} a second Lyapunov solve using $X_1(1)$ as input:
\begin{equation}
    X_1(1) = X_2 - \mathcal{I}_0 (X_2).
    \label{eq:Lyap2}
\end{equation}
Combining these relations with \cref{eq:moments_from_f} yields:
\bea
    \mu &= \tr(X_1),\\
    \sigma^2 &= -\mu (\mu + 1) + 2\tr(X_2),
    \label{eq:sigma_Lyap}
\eea
with $X_1, X_2$ obtained respectively from \cref{eq:Lyap1} at $z = 1$ and \cref{eq:Lyap2}.

The discrete-time cost function can then be written as
\begin{equation}
    \mathrm{CF}_d = 2d \tr (X_2) - (d + 1) \mudt^2
\end{equation}

In practice, one may eliminate $X_1$ and $X_2$ completely. 
For every admissible $K_0$ such that the resolvent exists, the matrices $X_1$ and $X_2$ are uniquely determined by the nested Lyapunov equations
\begin{equation}
    \varrho = X_1 - K_0 X_1 K_0^\dagger, 
    \qquad X_1 = X_2 - K_0 X_2 K_0^\dagger.
\end{equation}
Hence, both the mean and the cost function become functions of $K_0$ alone. 
The fixed-mean problem may then be written as
\begin{equation}
    \min_{K_0}\ \mathrm{CF}_d (K_0) 
    \qquad \text{subject to} 
    \qquad \mu(K_0) = \bar \mu, 
    \quad K_0^\dagger K_0 \le \id.
    \label{eq:reduced_fixed_mu_problem}
\end{equation}
This reduced formulation is particularly convenient for direct numerical optimization, because each evaluation of the objective and of the equality constraint only requires two discrete Lyapunov solves.

\subsection{Computing $\mu$ and $\sigma^2$ in Continuous Time}
\label{app:ct_lyapunov}
In continuous time, analogously as above, the moments of the tick distribution can be written directly in terms of the generator $\mathcal L_0$ of the evolution equation using \cref{eq:f(t)_CT_def}:
\begin{equation}
    \mu_n = (-1)^n n!\,\tr \left[\mathcal L_0^{-n} \varrho \right].
    \label{eq:ct_moment_formula}    
\end{equation}
We can determine the moments recursively by solving the continuous-time Lyapunov equation~\cite{Dost2023,Meier2025a},
\begin{equation}
    \varrho = \mathcal L_0 X_1^{(\rm ct)}, 
    \qquad X_1^{(\rm ct)} = \mathcal L_0 X_2^{(\rm ct)},
    \label{eq:ct_X12_def}
\end{equation}
first for $X_1^{(\rm ct)}$ and then for $X_2^{(\rm ct)}$.
Those are the continuous-time counterparts of the matrices $X_1$ and $X_2$ from \cref{eq:Lyap1,eq:Lyap2} in the discrete-time case.
Thus, the mean and variance of the tick time are then computed by exactly the same nested equations as in the discrete model:
\begin{equation}
    \mu = \tr\left[X_1^{(\rm ct)}\right],
    \qquad \sigma^2 = 2\tr \left[X_2^{(\rm ct)}\right] - \mu^2.
    \label{eq:ct_mu_sigma_X}
\end{equation}

\section{Parametrization and Numerical Structure of the Discrete-Time Ansatz}
\label{app:discrete_numerical_ansatz}
In this appendix, we give the explicit finite-dimensional parametrization used for the discrete-time optimization in \cref{subsec:numerical_optimization_discrete_time} and summarize the structure observed in the optimized solutions.
The ansatz combines coherent internal evolution with a pure, real initial state and a complex binary unsharp detector.
We first describe the reduced rank-one family used for the results in the main text. 
We then explain how this family emerges from broader searches over the detector effect, before discussing the additional simplifications and symmetries present for a qubit.

\subsection{Finite-dimensional rank-one parametrization}
\label{app:disc_opt_parametrization}
\textit{Coherent evolution.}---%
We work in the eigenbasis $\{\ket{k}\}_{k=0}^{d-1}$ of the unitary applied during one no-tick step. 
Equivalently, if the evolution is generated by a time-independent Hamiltonian $H\ket{k} = E_k \ket{k}$ over an interval $\tau_0$, then
\begin{equation}
    U_0 = e^{-iH\tau_0}
        =\ketbra{0} + \sum_{k=1}^{d-1}e^{iu_k}\ketbra{k},
    \label{eq:app_U0_parametrization}
\end{equation}
with $u_k = - \bigl(E_k - E_0 \bigr) \tau_0 \pmod{2\pi}$ being the variables used in the optimization. 
An overall phase of $U_0$ has no physical effect and has been removed in \cref{eq:app_U0_parametrization}. 
The coherent part of the ansatz therefore contributes $d - 1$ independent phases.
\begin{figure}[t]
    \centering 
    \includegraphics[width=\linewidth]{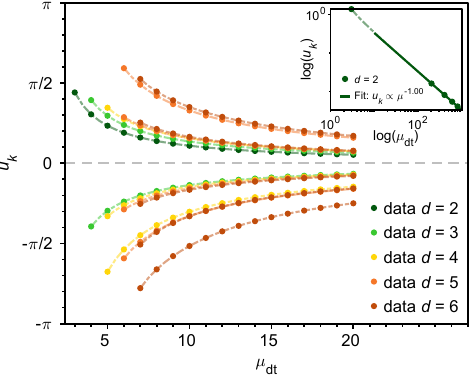}
    \caption{Numeric result of the optimal coherent evolution parameter $u_k$ scaling with $\mudt$ in a single run for dimensions $d=2,\dots,6$, (cf.~\cref{eq:app_U0_parametrization}).
    For each dimension, there are $d - 1$ parameters $u_k$, due to $k = 1, \dots, d-1$.
    The numerical results from several runs show that there is always a symmetry between the possible achievable values, meaning there are multiple equivalent optimal choices for the set $\{u_k\}_k$.
    Here, we show one such choice.
    In all instances $|u_k| \propto \mudt^{-1}$ in the limit of large mean $\mudt$, as exemplarily shown for $d = 2$ in the inset.}
    \label{fig:UniPhase}
\end{figure}

\textit{Initial state.---}%
We take the initial state to be pure,
\begin{equation}
    \varrho_{\mathrm{in}} = \ketbra{\psi}.
\end{equation}
In a $d$-dimensional Hilbert space, a generic normalized vector can be written in hyperspherical coordinates as
\begin{equation}
    \begin{split}
        \ket{\psi(\bm{\theta},\bm{\phi})}
        = {}& \sum_{k=0}^{d-2} e^{i\phi_k} \left( \prod_{j=1}^{k} \sin\theta_j \right) \cos\theta_{k+1} \ket{k}\\
            & + e^{i\phi_{d-1}} \left( \prod_{j=1}^{d-1} \sin\theta_j \right) \ket{d-1},
    \end{split}
    \label{eq:app_psi_d}
\end{equation}
where $\theta_j\in[0,\pi/2]$, $\phi_j\in[0,2\pi)$, the empty product equals one, and we fix the irrelevant global phase by setting $\phi_0 = 0$. 
Thus, a pure state is specified by $2(d-1)$ real parameters. 
The detector direction $\ket{\Psi}$ is parametrized in the same way, using an independent pair of angle vectors,
\begin{equation}
    \ket{\Psi} = \ket{\psi(\bm{\beta}, \bm{\gamma})}.
\end{equation}
Note by symmetry, one of the states (either the initial state or the detector direction) can be chosen to be real, thereby eliminating $d-1$ phase parameters from its description.

\textit{Rank-one detector and no-tick operation.}---%
As in \cref{eq:rank_one_tick_effect}, the tick effect is taken to be rank one,
\begin{equation}
    \begin{split}
        M_1 &= (1 - q) \ketbra{\Psi},\\
        M_0 &= \id - M_1 = q \ketbra{\Psi} + \bigl(\id - \ketbra{\Psi} \bigr).
    \end{split}
    \label{eq:app_rank_one_effects}
\end{equation}
Here $0 \leq q \leq 1$ is the no-tick eigenvalue along the detectable direction and $1 - q$ is the strength of a single detection attempt,
and we have
\begin{equation}
    \sqrt{M_0} = \sqrt q\,\ketbra{\Psi} + \bigl(\id - \ketbra{\Psi} \bigr)
               = \id - (1 - \sqrt q) \ketbra{\Psi} .
    \label{eq:app_sqrt_M0}
\end{equation}

Thus, the overall no-tick Kraus operator and operation are
\begin{equation}
    K_0=U_0\sqrt{M_0},
    \qquad
    \mathcal I_0(\varrho)=K_0\varrho K_0^\dagger.
    \label{eq:app_parametrized_notick_operation}
\end{equation}

For the rank-one effect in \cref{eq:app_rank_one_effects}, the discrete hazard rate reduces to the evolving overlap with a single detector direction,
\begin{equation}
    h_L = 1-\frac{f_L}{f_{L-1}}
        = \tr\!\left[M_1\varrho_{L-1}\right]
        = (1 - q) \bra{\Psi} \varrho_{L-1} \ket{\Psi}.
    \label{eq:app_hazard_parametrized}
\end{equation}

In the end, the reduced ansatz contains
\begin{equation}
    (d - 1) + 3(d - 1) + 1 = 4d - 3
    \label{eq:app_parameter_count}
\end{equation}
real variational parameters: the $d-1$ relative eigenphases of $U_0$, the $3(d - 1)$ parameters for a generic real pure state and complex detection direction, and the attenuation parameter $q$.
Although the chosen basis diagonalizes $U_0$, no restriction is thereby imposed on the relative orientation of the initial state and the detector, since both are varied freely in that basis.

Figures~\ref{fig:UniPhase} and~\ref{fig:MeasurStrngth} show the numerical results of the coherent evolution parameters $u_k$ and the measurement strength
$1-q$ respectively, showing how they vary as a function of the mean tick time $\mudt$, in particular with an inverse linear scaling for large
$\mudt$.

\subsection{Emergence of the rank-one detector}
\label{app:disc_opt_rank_one_emergence}
The rank-one form  of detection in \cref{eq:app_rank_one_effects} can also be motivated directly from the broader numerical searches. 
To test the spectral structure of the detector, consider first a general no-tick effect written in its eigenbasis,
\begin{equation}
    M_0 = \sum_{i=1}^{d} q_i \ketbra{\Psi_i}, 
    \qquad 0 \leq q_i \leq 1,
    \label{eq:app_general_spectral_effect}
\end{equation}
where $\ket{\Psi_i}$, for $i = 1, \dots, d$ is an orthonormal basis. 
For low dimensions ($d=2,3$ and $4$), we investigated the optimal such higher rank detectors.
In all those instances, the optimization over the detector eigenvalues consistently drove the solutions to the boundary pattern
\begin{equation}
    q_i =
    \begin{cases}
        q(\bar\mu), & i = k,\\
        1,          & i \neq k,
    \end{cases}
    \qquad k \in \{1, \ldots, d\}.
    \label{eq:app_rank_one_spectrum_observation}
\end{equation}
Thus, for the optimization in higher dimensions, we use
\begin{equation}
    \begin{split}
        M_1 &= \Big(1-q(\bar\mu)\Big)\ketbra{\Psi_k},\\
        M_0 &= \id - \Big(1-q(\bar\mu)\Big)\ketbra{\Psi_k},
    \end{split}
\end{equation}
which is precisely the reduced ansatz from \cref{eq:app_rank_one_effects}. 
The choice of the index $k$ carries no physical significance, since permutations of the measurement eigenbasis produce equivalent parametrizations. 
This observation is numerical rather than an analytic proof that rank-one detection is globally optimal among all quantum instruments; it is the empirical justification for the restricted search used in the main text.

An intermediate parametrization, useful in the qubit search and in tests with an isotropic orthogonal subspace, is
\begin{equation}
    M_0 = q_{\parallel} \ketbra{\Psi} + q_{\perp} \bigl(\id - \ketbra{\Psi} \bigr),
    \qquad q_{\parallel}, q_{\perp} \in [0,1].
    \label{eq:app_two_eigenvalue_effect}
\end{equation}
Before imposing the rank-one structure, this family has  $4d - 2$ real parameters. 
The branch $q_{\perp} = 1$ and $q_{\parallel} = q$ reduces to \cref{eq:app_rank_one_effects}. 
For $d > 2$, the opposite boundary branch $q_{\parallel} = 1$, $q_{\perp} < 1$ would instead give a rank-$(d - 1)$ tick effect and should not be identified with the rank-one branch. 
For a qubit, however, both eigenspaces are one dimensional and the two branches are equivalent after relabelling the detector direction.

\subsection{Large-mean regime and continuous-time scaling}
\label{app:disc_opt_large_mean}

\begin{figure}[t]
    \centering
    \includegraphics[width=\linewidth]{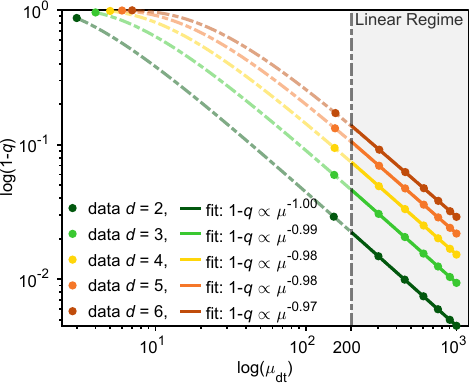}
    \caption{Numerically optimized detector strength \(1-q\) as a
    function of the prescribed mean tick time \(\bar\mu\), for
    dimensions \(d=2,\ldots,6\). In the large-\(\bar\mu\) regime, the
    data approach the inverse scaling \(1-q\propto\bar\mu^{-1}\)
    expected on the branch converging to a continuous-time rank-one
    detector.}
    \label{fig:MeasurStrngth}
\end{figure}

As the prescribed mean \(\bar\mu\) increases, the optimized detector
becomes progressively weaker. Numerically, its strength obeys
\begin{equation}
    1-q(\bar\mu)\propto\bar\mu^{-1},
    \label{eq:app_detector_strength_scaling}
\end{equation}
in agreement with \cref{eq:detector_strength_scaling}, as shown in
\cref{fig:MeasurStrngth}.

The coherent evolution displays a corresponding weak-step scaling.
Because the eigenphases \(u_k\) are defined only modulo \(2\pi\), let
\(\widetilde u_k(\bar\mu)\in\mathbb{R}\) denote representatives chosen
continuously along one optimal branch and such that
\(\widetilde u_k(\bar\mu)\to0\) as \(\bar\mu\to\infty\). After fixing
the irrelevant common phase, setting \(\widetilde u_0=0\), and
reordering the energy eigenvectors, the numerical solutions in
\cref{fig:UniPhase} are, for large $\mudt$, approximately described by $\widetilde u_k(\bar\mu)=O(\bar\mu^{-1})$ with
\begin{equation}
    \widetilde u_{k+1}(\bar\mu)-\widetilde u_k(\bar\mu)
    \simeq \frac{\Delta_d}{\bar\mu},
    \label{eq:app_unitary_phase_scaling}
\end{equation}
where \(\Delta_d\) is a dimension-dependent spacing and $k=0,\ldots,d-2$. Thus, we see empirically that the
rescaled phases \(\bar\mu\,\widetilde u_k\) approach, to numerical
accuracy, an approximately arithmetic pattern. 

The simultaneous scalings in
\cref{eq:app_detector_strength_scaling,eq:app_unitary_phase_scaling}
are the discrete signature of the continuous-time limit. Introduce a
step duration \(\delta\), keep the physical mean
\(\mu_*=\delta\bar\mu\) finite, and write
\begin{equation}
    U_0=e^{-iH\delta},
    \qquad
    q=1-\gamma\delta+O(\delta^2).
\end{equation}
For the phase representatives that vanish with \(\delta\),
\begin{equation}
    \widetilde u_k
    =-\bigl(E_k-E_0\bigr)\delta .
    \label{eq:app_phase_energy_continuous}
\end{equation}
Consequently, the approximate arithmetic pattern of the rescaled
phases corresponds, up to an overall energy shift and time rescaling,
to an approximately equally spaced spectrum of the limiting
Hamiltonian, which we also observe directly in the continuous limit, cf.~\cref{fig:Fig_OptCT}(a). 
Moreover, using \cref{eq:app_sqrt_M0}, the no-tick Kraus
operator satisfies
\begin{equation}
    \begin{split}
        K_0
        &=U_0\sqrt{M_0}\\
        &=\id-iH\delta
          -\frac{\gamma\delta}{2}\ketbra{\Psi}
          +O(\delta^2).
    \end{split}
    \label{eq:app_K0_continuous_expansion}
\end{equation}
Thus, keeping a finite physical mean while
\(\bar\mu\sim\delta^{-1}\) requires both
\(1-q=O(\bar\mu^{-1})\) and
\(\widetilde u_k=O(\bar\mu^{-1})\).
Equation~\eqref{eq:app_K0_continuous_expansion} connects the optimized
discrete ansatz directly to the non-Hermitian no-tick generator
parametrized in \cref{appendix:parametrization_ct_model}.

The empirical large-mean behavior of the cost can be summarized by
writing, for fixed \(d\),
\begin{equation}
    \mathrm{CF}^{\mathrm{opt}}_d(\bar\mu)
    =-c_d\bar\mu^2+O(\bar\mu).
    \label{eq:app_cost_large_mean}
\end{equation}
Using the identities in \cref{eq:sigma_Lyap} together with the
definition of the cost function gives
\begin{equation}
    \begin{split}
        \sigmadt^2
        &=
        \frac{1-c_d}{d}\,\bar\mu^2+O(\bar\mu),\\
        \tr(X_2)
        &=
        \frac{d+1-c_d}{2d}\,\bar\mu^2+O(\bar\mu).
    \end{split}
    \label{eq:app_asymptotic_linked_fits}
\end{equation}
Hence the leading quadratic coefficients of the cost, variance, and
survival weight are fixed by the single number \(c_d\). Together with
\cref{eq:app_detector_strength_scaling,eq:app_unitary_phase_scaling},
this captures the relevant numerical trend: both the coherent rotation
and the detection probability per step vanish as
\(O(\bar\mu^{-1})\), while their accumulation over
\(O(\bar\mu)\) steps produces a finite conditional dynamics that
concentrates the tick within a narrow relative time window.

\subsection{Qubit specialization and symmetries}
\label{subsec:qubit_case}
For $d = 2$, it is convenient to replace the hyperspherical angles by the usual Bloch-sphere polar angles. 
We write
\begin{equation}
    \begin{split}
        \ket{\psi} &= \cos\!\left(\frac{\theta_{\psi}}{2}\right)\ket{0} + \sin\!\left(\frac{\theta_{\psi}}{2}\right)\ket{1},\\
        \ket{\Psi} &= \cos\!\left(\frac{\theta_{\Psi}}{2}\right)\ket{0} + e^{i\phi_{\Psi}} \sin\!\left(\frac{\theta_{\Psi}}{2}\right)\ket{1},
    \end{split}
    \label{eq:app_qubit_states}
\end{equation}
with
\(\theta_{\psi},\theta_{\Psi}\in[0,\pi]\)
and
\(\phi_{\Psi}\in[0,2\pi)\). For the unitary we write
\begin{equation}
    U_0 = \ketbra{0} + e^{iu} \ketbra{1}.
    \label{eq:app_qubit_unitary}
\end{equation}
The Bloch polar angles in \cref{eq:app_qubit_states} are twice the corresponding hyperspherical angles in \cref{eq:app_psi_d}. 
The reduced qubit ansatz is therefore described by the five parameters
\begin{equation}
    \bm\xi_5 = \left( \theta_{\psi}; \theta_{\Psi}, \phi_{\Psi}; u; q \right).
    \label{eq:app_qubit_parameter_vector}
\end{equation}

For completeness, the two-eigenvalue detector \cref{eq:app_two_eigenvalue_effect} has a familiar Bloch-vector form. 
With $\ketbra{\Psi} = (\id + \hat{\bm n} \cdot \bm\sigma)/2$, one finds
\begin{equation}
    M_0 =\frac{1}{2} \left[ (q_{\parallel} + q_{\perp}) \id + (q_{\parallel} - q_{\perp}) \hat{\bm n} \cdot \bm\sigma \right].
    \label{eq:app_qubit_bloch_effect}
\end{equation}
Comparing this with $M_0 = [(1 + b) \id + \eta\,\hat{\bm m} \cdot \bm\sigma]/2$ gives
\begin{equation}\label{eq:appbetaabsb}
    b = q_{\parallel} + q_{\perp} - 1, 
    \qquad \eta = \lvert q_{\parallel} - q_{\perp}\rvert, 
    \qquad \lvert b\rvert \leq 1 - \eta ,
\end{equation}
with
\(\hat{\bm m}
=\operatorname{sgn}(q_{\parallel}-q_{\perp})\hat{\bm n}\).
The last condition in \cref{eq:appbetaabsb} is automatically satisfied when $q_{\parallel}, q_{\perp} \in [0,1]$. 
The optimized boundary solution, for which one eigenvalue equals one, saturates this positivity inequality. 
The unrestricted qubit search based on $q_{\parallel}$ and $q_{\perp}$ contains six parameters; after exploiting the observed boundary structure it reduces to \cref{eq:app_qubit_parameter_vector}.

The qubit optimizations also display several equivalent branches. 
Numerically, the optimized solutions satisfy
\(\theta_\psi=\theta_\Psi=\pi/2\).
Finally, complex conjugation maps
\begin{equation}
    (u,\phi_\Psi) \longmapsto (-u, -\phi_\Psi) \pmod{2\pi}
\end{equation}
without changing any first-tick probability. 
This accounts for the conjugate optimal branches that are often displayed numerically as $u \leftrightarrow 2\pi - u$.

The two-parameter qubit clock analyzed in Ref.~\cite{Budroni2021Ticking-clock} is contained in this parametrization as a restricted subfamily. 
Its analytic cost,
\begin{equation}
    \mathrm{CF}_{2}^{(\mathrm{ref})}(\mudt) = -\frac{\mudt^2 (\mudt - 2)^2}{2  (\mudt - 1)^2}, 
    \qquad \mudt > 2,
    \label{eq:app_reference_qubit_cost}
\end{equation}
provides a useful check of the numerical implementation. 
It should not, however, be identified with the optimum over the full ansatz: the additional degrees of freedom retained here yield the stronger asymptotic qubit optimum derived in \cref{subsec_analytical_d2}.

\section{Analytical Bound for the Two-Dimensional Case}
\label{subsec_analytical_d2}
To validate our numerical findings, the optimal counting process can be solved analytically for the two-dimensional case ($d = 2$). 
Once again, we model the discrete quantum counting process using a single Kraus operator $K_0 = U_0 \sqrt{M_0}$ with modulus $M_0 = \operatorname{diag}(1, q)$. 
Any such single-Kraus no-tick operation can be transformed into an upper-triangular Schur form $T = Z^\dagger K_0 Z$ via a unitary matrix $Z$. 
Because we consider an arbitrary pure initial state, this unitary transformation simply corresponds to a change of basis that can be absorbed into the state preparation, effectively allowing us to set $Z = \id$. 

The mean waiting time $\mudt$ and the discrete-time cost function $\mathrm{CF}_2$ are then completely determined by the traces of the solutions to the nested discrete Lyapunov equations
\begin{equation}
    X_1 - T X_1 T^\dagger = \tilde{\varrho}, 
    \qquad X_2 - T X_2 T^\dagger = X_1
\end{equation}
yielding $\mudt = \tr(X_1)$ and $\mathrm{CF}_2 = 4\,\tr(X_2) - 3 \mudt^2$, where $\tilde{\varrho}$ is the initial pure state in the Schur basis, parameterized by a population angle $\alpha$ and a relative phase $\delta$.

Fixing the global phase, the invariant determinant and Frobenius norm of the physical contraction $T$ strictly constrain its non-normal coupling. 
This allows us to parameterize the operator exactly as a function of its hyperspherical loss angles $\gamma_1, \gamma_2$ and internal coherent phases $\phi, \psi$
\begin{equation}
    T = 
    \begin{pmatrix} 
        \cos\gamma_1 & \sin\gamma_1 \sin\gamma_2 \,e^{i\psi}\\ 
        0            & \cos\gamma_2 \,e^{i\phi}
    \end{pmatrix}.
\end{equation}
Alongside this operator, the arbitrary initial pure state is explicitly defined in the Schur basis as
\begin{equation}
    \tilde{\varrho} = 
    \begin{pmatrix} 
        \cos^2 \alpha & \cos\alpha \sin\alpha \,e^{-i\delta}\\ 
        \cos\alpha \sin\alpha \,e^{i\delta} & \sin^2 \alpha 
    \end{pmatrix}.
\end{equation}
In the evaluation of both traces, the phases $\psi$ and $\delta$ always appear together as $\psi + \delta$. 
For this reason, we can drop $\delta$ without loss of generality. 
In this representation, the measurement attenuation parameter is precisely given by $q = \cos^2 \gamma_1 \cos^2 \gamma_2$.

While the exact traces of $X_1$ and $X_2$ can be obtained analytically for any arbitrary parameters, accessing the large mean waiting time regime ($\mudt \gg 1$) requires the attenuation parameters to asymptotically approach zero. 
By assuming a symmetric attenuation $\gamma_1 = \gamma_2 = \gamma$, we can expand the analytical traces to leading order in $\gamma$. 
This expansion reveals a critical structural requirement to access the optimal regime where $\mathrm{CF}_2 < 0$. 
One can show that for the optimal protocol, the relative coherent phase $\phi$ must scale quadratically with the measurement attenuation as $\phi = c \gamma^2$.

\begin{figure*}
    \centering
    \includegraphics[width=\linewidth]{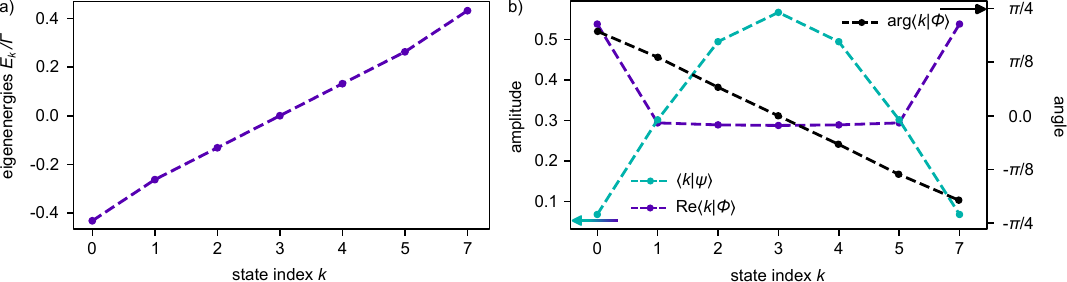}
    \caption{Optimal parameters for continuous-time for $d = 7$.
    a) The energy eigenvalues are shown in units of the tick decay rate.
    b) Overlap of the initial state and tick state with the energy eigenstates are shown on the left scale.
    The complex argument (in radian) is shown on the right scale.}
    \label{fig:Fig_OptCT}
\end{figure*}

Keeping only the leading order terms, the traces factor into scale-invariant numerators
\begin{equation}
    \tr(X_1) \approx \frac{N_1 (\alpha, \psi, c)}{\gamma^2}, 
    \qquad \tr(X_2) \approx \frac{N_2 (\alpha, \psi, c)}{\gamma^4}
\end{equation}
where the analytical coefficients evaluate to
\begin{align}
    N_1(\alpha, \psi, c) &= 1 + \frac{2\sin^2\alpha}{1 + c^2} + \frac{\cos\psi - c \sin\psi}{1 + c^2} \sin(2\alpha)\\
    N_2(\alpha, \psi, c) &= 1 + \frac{2(3 + c^2)}{(1 + c^2)^2} \sin^2\alpha \nonumber\\
    &+ \frac{2\cos\psi - c(3 + c^2) \sin\psi}{(1 + c^2)^2} \sin(2\alpha)
\end{align}

So we can optimize the scale-invariant ratio $\mathrm{CF}_2 / \mudt^2$. The attenuation parameter $\gamma$ completely drops out, reducing the problem to a purely algebraic minimization over the internal parameters
\begin{equation}
    \lim_{\mudt \to \infty} \left( \frac{\mathrm{CF}_2}{\mudt^2} \right) 
    = \frac{4 N_2 (\alpha, \psi, c)}{\big(N_1 (\alpha, \psi, c)\big)^2} - 3
\end{equation}

Minimizing this ratio yields a global asymptotic bound of exactly $-0.5982$. 
This minimum is attained at $\alpha \approx 55.67^\circ$, $\psi \approx 23.87^\circ$, and a phase-scaling coefficient of $c \approx -1.0359$. 
This strictly confirms the quantum advantage over the classical finite-memory benchmark in two dimensions and perfectly corroborates the asymptotic scaling observed in our numerical optimization.

\section{Parametrization of the Continuous-Time Model}
\label{appendix:parametrization_ct_model}
In this appendix, we discuss how the continuous-time counting process is parameterized for the numerical optimization of precision $\mathcal N$.

For this purpose, the model can be parametrized in the energy eigenbasis. 
That is, we may assume $H$ to be diagonal with the eigenvalues $E_1, \dots, E_d \in \mathbb R$ as free parameters, where $d$ is the Hilbert space dimension.
The initial state of the clock can be parametrized in the energy eigenbasis as
\begin{equation}
    \ket{\psi} = \sum_{i=1}^d \psi_i \ket{i}.
\end{equation}
We may assume that all $\psi_i \in \mathbb R$ are real.
Otherwise, we could always write $\ket{\psi} = e^{iD} |\tilde\psi \rangle$, with $\tilde\psi_i \in \mathbb R$ and $D$ a diagonal matrix in the energy eigenbasis of $H$.
The transformed evolution generator then reads
\begin{equation}
    e^{iD} \left(H - \frac{i}{2} \ketbra{\Phi}{\Phi} \right) e^{-iD} 
    = H - \frac{i}{2} e^{iD} \ketbra{\Phi}{\Phi} e^{-iD}.
    \label{eq:phases}
\end{equation}
Hence, we can fix the initial state to be purely real, $\psi_i \in \mathbb R$ at the cost of the tick state $\ket{\Phi}$ in general being complex.

Let us thus parametrize the tick state with
\begin{equation}
    \ket{\Phi} = \sum_{i=1}^d \Phi_i e^{i\theta_i} \ket{i}.
\end{equation}
The parameters are real, i.e., $\Phi_i, \theta_i \in \mathbb R$.
Moreover, we may assume without loss of generality that $\theta_1 = 0$, as we can always subtract global phases, additionally to the symmetry in \cref{eq:phases}.
This also implies that we can shift the evolution generator by an arbitrary constant; here, we fix $\sum_{i=1}^d E_i = 0$.

Including normalization constraints, $1 = \sum_{i=1}^d \Phi_i^2$ for the tick state, $1 = \sum_{i=1}^d \psi_i^2$ for the initial state, the $d$-dimensional problem thus has in total
\begin{equation}
    4(d - 1) \text{ real parameters } \Phi_i, \theta_i, \psi_i, E_i
\end{equation}
Note that, in principle, there would be $+1$ additional parameter that sets the overall time scale of the problem, and thus the scale of $\muct$ and $\sigmact$.
This additional parameter drops out in continuous time, however, due to the time-scale invariance of $\mathcal N$, resulting in one less free parameter than the discrete-time problem.
In Fig.~\ref{fig:Fig_OptCT}, we show the parameters obtained from numerically optimizing $\mathcal N$ for the exemplary case of $d = 7$.
The optimal parameters qualitatively resemble the parameters proposed in~\cite{Woods2022,Dost2023}.

\bibliography{ref.bib}

\end{document}